\documentclass{amsart} 

\usepackage{hyperref}
\hypersetup{
  colorlinks   = true, 
  urlcolor     = blue, 
  linkcolor    = blue, 
  citecolor   = red 
}

\usepackage{amssymb}
\usepackage[foot]{amsaddr}
\usepackage{array}

\usepackage[normalem]{ulem}
\usepackage[margin=2.7cm]{geometry} 
\usepackage{bm}
\usepackage{csquotes}
\usepackage{mathtools}
\usepackage{enumerate}
\usepackage[dvipsnames]{xcolor}
\usepackage{pagecolor}
\usepackage{tikz}
\usepackage{pgfplots}
\pgfplotsset{compat=1.14}

\newtheorem{theorem}{Theorem}[]

\newtheorem{proposition}[theorem]{Proposition}
\newtheorem{corollary}[theorem]{Corollary}
\newtheorem{lemma}[theorem]{Lemma}

\theoremstyle{definition}
\newtheorem{definition}[theorem]{Definition}
\newtheorem{example}[theorem]{Example}
\newtheorem{remark}[theorem]{Remark}

\newcommand{\C}{\mathcal{C}}
\newcommand{\A}{\mathcal{A}}
\newcommand{\B}{\mathcal{B}}

\newcommand{\calK}{\mathcal{K}}

\newcommand{\Short}{\mathrm{Short}}

\newcommand{\F}{\mathbb{F}}
\newcommand{\Gal}{\mathrm{Gal}}
\newcommand{\N}{\mathbb{N}}
\newcommand{\rk}{\mathrm{rk}}

\newcommand{\End}{\mathrm{End}}

\newcommand{\Ext}{\mathrm{Ext}}

\newcommand{\Tr}{\mathrm{Tr}_{\LL/\K}}
\newcommand{\Hom}{\mathrm{Hom}}
\newcommand{\Span}[2]{\mathbf{Span}_{#1}\left\{#2\right\}}
\newcommand{\Fq}{\mathbb{F}_{q}}
\newcommand{\RM}{\mathrm{RM}}
\newcommand{\dist}{\mathrm{d}}
\newcommand{\Id}{\textrm{Id}}

\newcommand{\w}{\mathrm{wt}}

\newcommand{\Z}{\mathbb{Z}}
\newcommand{\Q}{\mathbb{Q}}
\newcommand{\LL}{\mathbb{L}}
\newcommand{\K}{\mathbb{K}}
\newcommand{\E}{\mathbb{E}}

\DeclareMathOperator{\ev}{ev}
\DeclareMathOperator{\supp}{Supp}
\newcommand{\evB}{\ev_{\B}}
\newcommand\mydef{\coloneqq}

\newcommand\bfa{{\bm a}}
\newcommand\bfb{{\bm b}}

\newcommand\bfi{{\bm i}}
\newcommand\bfj{{\bm j}}
\newcommand\bfu{{\bm u}}
\newcommand\bfv{{\bm v}}
\newcommand\bfn{{\bm n}}
\newcommand\bfx{{\bm x}}
\newcommand\bfy{{\bm y}}
\newcommand\bfone{{\bm 1}}

\newcommand\bftheta{{\bm \theta}}
\newcommand\bfalpha{{\bm \alpha}}
\newcommand\bfzeta{{\bm \zeta}}
\newcommand{\bb}{b}

\newcommand{\uuu}{u}
\newcommand{\vvv}{v}

\newcommand{\aalpha}{\alpha}
\newcommand{\bbeta}{\beta}
\newcommand{\llambda}{\lambda}

\newcommand{\map}[4]{
  \left\{
    \begin{array}{ccc}
      #1 & \longrightarrow & #2 \\
      #3 & \longmapsto     & #4
    \end{array}
  \right.
}

\DeclareMathOperator{\Ann}{Ann}

\definecolor{light-gray}{gray}{0.90}

\renewcommand{\leq}{\leqslant}
\renewcommand{\geq}{\geqslant}
\renewcommand{\le}{\leqslant}
\renewcommand{\ge}{\geqslant}

\title[Rank-metric codes over arbitrary Galois extension]{Rank-metric codes over arbitrary Galois extensions and rank analogues of Reed--Muller codes}

\author{Daniel Augot}
\email{daniel.augot@inria.fr}

\author{Alain Couvreur}
\address{Inria \& LIX, CNRS UMR 7161, \'Ecole Polytechnique, 91120 Palaiseau Cedex, France}
\email{alain.couvreur@lix.polytechnique.fr}

\author{Julien Lavauzelle}
\address{Univ. Rennes, CNRS, IRMAR -- UMR 6625, F-35000 Rennes, France}
\email{julien.lavauzelle@univ-rennes1.fr}

\author{Alessandro Neri}
\address{Institute for Communications Engineering, TU Munich, Germany}
\email{alessandro.neri@tum.de}

\date{}

\begin{document}

\maketitle
\thispagestyle{empty}

\begin{abstract}
  
  This paper extends the study of rank-metric codes in extension
  fields $\mathbb{L}$ equipped with an arbitrary Galois group
  $G = \mathrm{Gal}(\mathbb{L}/\mathbb{K})$. We propose a framework
  for studying these codes as subspaces of the group algebra
  $\mathbb{L}[G]$, and we relate this point of view with usual notions
  of rank-metric codes in $\mathbb{L}^N$ or in
  $\mathbb{K}^{N\times N}$, where $N = [\mathbb{L} : \mathbb{K}]$. We
  then adapt the notion of error-correcting pairs to this context, in
  order to provide a non-trivial decoding algorithm for these codes.

  We then focus on the case where $G$ is abelian, which leads us to
  see codewords as elements of a multivariate skew polynomial ring. We
  prove that we can bound the dimension of the vector space of zeroes
  of these polynomials, depending of their degree. This result can be
  seen as an analogue of Alon--F{\" u}redi theorem --- and by means,
  of Schwartz--Zippel lemma --- in the rank metric. Finally, we
  construct the counterparts of Reed--Muller codes in the rank metric,
  and we give their parameters. We also show the connection between
  these codes and classical Reed--Muller codes in the case where
  $\mathbb{L}$ is a Kummer extension.
\end{abstract}


\section{Introduction}

\subsection{Context.}
Rank-metric codes were introduced independently by Delsarte in
\cite{de78} and Gabidulin in \cite{ga85a} for combinatorial
purposes. Roth rediscovered them in~\cite{ro91} and showed their
application to crisscross error-correction. In the same year,
Gabidulin, Paramonov and Tretjakov proposed the use of rank-metric
codes for cryptographic purposes, designing the GPT
cryptosystem~\cite{gabidulin1991ideals}. More recently, Silva, Koetter
and Kschischang showed how these codes can be used in network
coding~\cite{si08j}. This series of papers raised the interest of many
researchers from different areas, who investigated their mathematical
properties and further applications.

Rank-metric codes have been introduced as spaces of $N \times M$
matrices over a finite field $\Fq$ by Delsarte, while Gabidulin
considered them as $\F_q$-linear spaces of vectors of length $N$ over
an extension field $\F_{q^M}$. The two representations are equivalent:
when choosing an $\Fq$-basis of $\F_{q^N}$, one can write each element
of $\F_{q^N}$ as a column of its coordinates in this basis. Thus, the
\emph{rank distance} on $\Fq^{N \times M}$, defined as the rank of the
difference of two matrices, is equivalent to the distance on
$\F_{q^N}^M$ defined as the rank of the difference of the matrix
representations of two vectors.

In the case $M=N$ it is also possible to view matrices as
endomorphisms. More precisely, one has
\[\Fq^{N \times N} \cong \End_{\Fq}(\F_{q^N}) \cong \mathcal
L[x]/(x^{q^N}-x),\] where $\mathcal L[x]$ is the ring of
$q$-polynomials with coefficients in $\F_{q^N}$ endowed with addition
and composition, and $(x^{q^N}-x)$ denotes the two--sided
ideal spanned by $x^{q^N}-x \in \mathcal{L}[x]$.
Recall that a \emph{$q$-polynomial} (or \emph{linearized polynomial})
is an element $P(x) \in \F_{q^N}[x]$ such that the exponents of
monomials involved in $P$ are powers of $q$. Moreover, the matrix
algebra $\Fq^{N \times N}$ is also isomorphic to the skew group
algebra $\F_{q^N}[G]$, where $G=\Gal(\F_{q^N}/\Fq)$, endowed with the
usual addition and the multiplication defined by the rule
\[\forall\, g_i, g_j \in G, \bb_i, \bb_j \in \F_{q^N}, \quad
  (\bb_ig_i)\circ (\bb_jg_j)=(\bb_i g_i(\bb_j))(g_i \circ g_j).
\]
We refer to \cite{wu2013linearized} for a complete
presentation of these equivalent representations.

The above isomorphisms make it easier to study the algebraic structure of
rank-metric codes, and have been used for designing rank-metric codes
with good parameters. This is the case of the well-known family of
\emph{Gabidulin codes}~\cite{de78, ga85a}. They were first defined as
the subspace of linearized polynomials of degree at most $q^{k-1}$,
which corresponds to
$\Span{\F_{q^N}}{ \sigma^i \mid i=0,\ldots, k-1}\subseteq
\F_{q^N}[G]$, where $\sigma$ is the $q$-Frobenius automorphism. This
family has been then generalized by Kshevetskiy and Gabidulin in
\cite{ks05}, by taking the subspace
$\Span{\F_{q^N}}{ \theta^i \mid i=0,\ldots, k-1}\subseteq
\F_{q^N}[G]$, where $\theta$ is any generator of the Galois group $G$.

This point of view was crucial for generalizing Gabidulin codes over
arbitrary cyclic Galois extensions. In a series of papers, Augot,
Loidreau and Robert~\cite{augot2013rank, augot2014generalization,
  augot2018generalized} investigated on the case where
$G \mydef \langle \theta \rangle$ is the
Galois group of a degree $N$ cyclic extension $\LL/\K$ (see also \cite[Section VI]{roth1996tensor}). The same ring
isomorphisms hold between $\K^{N \times N}$, $\End_{\K}(\LL)$ and the
skew group algebra $\LL[G]=\LL[\theta]$, and hence one can define a
Gabidulin code as the $\LL$-subspace in $\LL[G]$ generated by
$\theta^i$ for $i=0,\ldots, k-1$.

Gabidulin codes are considered as analogues in the rank metric of
Reed--Solomon codes. Indeed,
Reed--Solomon codes are obtained by considering the $\Fq$-subspace
$\Span{\F_q}{ x^i \mid i=0,\ldots, k-1} \subseteq \Fq[x]$. The
analogy can also be seen via their generator matrices. For
Reed--Solomon codes, the evaluation of the monomials $x^i$'s on a
subset of $\Fq$ yields a Vandermonde matrix, while for Gabidulin codes
the Moore matrix is obtained by the action of the $\theta^i$'s on a
$\K$-linearly independent subset of $\LL/\K$. Another analogy can be
found by studying the systematic generator matrices, which produces
Cauchy matrices for Reed--Solomon codes, and their $q$-analogue for
Gabidulin codes~\cite{ne18sys}.

Central to current research trends is the idea of finding
constructions in the Hamming metric that have a counterpart in the
rank metric, in order to obtain analogous objects. For instance, a
problem is whether one can construct Reed--Muller type codes for the
rank metric. Recall that $q$-ary Reed--Muller codes in $m$ variables
are obtained by considering the $\Fq$-subspace
$\Span{\Fq}{ x_1^{i_1}\cdots x_m^{i_m} \mid i_1+ \cdots + i_m
\leq r } \subseteq \Fq[x_1,\ldots,x_m]$ for a certain degree
$r$, and then evaluating all the polynomials in this subspace in every
point of $\Fq^m$. In order to obtain the same analogy as the one
between Gabidulin and Reed--Solomon codes, one should construct $m$
distinct automorphisms $\theta_1, \ldots, \theta_m\in G=\Gal(\LL/\K)$
which commute and span disjoints subgroups of $G$ of order $n$, and
then define the space
$$\mathrm{RM}_{\LL/\K}(r,n,m) \mydef \Span{\LL}{\theta_1^{i_1}
\circ \cdots \circ \theta_m^{i_m} \mid i_1+ \cdots +
i_m \leq r}.$$
This notably requires that $G$ contains a subgroup isomorphic to
$(\Z/n\Z)^m$.

In the finite field setting, Galois groups are cyclic. This explains
why up to now, no one succeeded in constructing Reed--Muller codes for
the rank metric that share the parameters of classical Reed--Muller
codes. Indeed, if one tries to get a subspace of $\F_{q^N}[G]$ of the
form $\mathrm{RM}_{\LL/\K}(r,n,m)$, then one has to choose the
$\theta_i$'s as powers of the same generator $\theta$, obtaining 
a generalized Gabidulin code or, more generally, a rank-metric code 
satisfying a Roos-like bound \cite{martinez2017roots, alfarano2020roos}.

\subsection{Overview.}
Motivated by this intuition, in this paper we study the general theory
of rank-metric codes over \emph{arbitrary} Galois extensions. We first
investigate the isomorphisms
$\K^{N \times N} \cong \End_{\K}(\LL)\cong \LL[G]$, showing equivalent
definitions of the rank metric. This also allows us to define the
counterparts of Moore matrices and Dickson matrices for general Galois
extensions, which are fundamental objects in order to determine the
rank of a linearized polynomial.
We prove that the definitions of
these matrices are consistent with the finite field case, and they
have exactly the same properties.
We then adapt the notion of error-correcting pairs to the context of codes
in $\LL[G]$. Error-correcting pairs were originally introduced by
Pellikaan~\cite{Pellikaan92}, and a rank-metric version was recently
proposed by Mart{\'{\i}}nez{-}Pe{\~{n}}as and Pellikaan~\cite{Martinez-PenasP17}.

Once developed the general theory of codes in $\LL[G]$ for arbitrary
finite groups $G$, we restrict to the case of abelian groups, which was the main
motivation of our project. In this context, elements of the group algebra
can be seen as skew polynomials in $\theta_1, \dots, \theta_m$, where
$G = \langle \theta_1, \dots, \theta_m \rangle$. We prove an upper bound
on the dimension of their space of zeros, depending on their degree.
This result can be seen as an analogue of Alon--F{\"u}redi theorem and
Schwartz--Zippel lemma in the rank metric setting.

We then naturally define
$\theta$-Reed--Muller codes as mentioned before, and study their
parameters. It turns out that this construction produces rank-metric
codes with the same parameters as $q$-ary Reed--Muller codes. Furthemore,
when restricting to Kummer extensions with Galois
group $G \simeq \Z/n_1\Z \times \cdots \times \Z/{n_m}\Z$, 
the $\theta$-Reed--Muller code shares the
structure of an \emph{affine variety code} or \emph{affine cartesian code}
(see \cite{geil2013weighted,lopez2014affine}).

Notice that in~\cite{GeiselmannU19}, Geiselmann and Ulmer also proposed
a generalisation of Reed--Muller codes by using skew polynomial rings.
However their work significantly differs from ours, since they use
iterated rings with non-trivial derivation in order to stand out
from classical Reed--Muller codes.

\subsection{Organisation.}
The paper is structured as follows. In Section~\ref{sec:preliminaries}
we recall basic notions in algebra that are useful to define the rank
metric on arbitrary Galois extension fields (Section~\ref{sec:rank}).
Dickson matrices are introduced in Section~\ref{sec:Dickson} where we
also determine their algebraic properties. We are then able to define
and describe the properties of rank-metric codes in $\LL[G]$ in
Section~\ref{sec:RMC} and their error-correcting pairs in
Section~\ref{sec:ECP}.  Next, Section~\ref{sec:thetapoly} is
dedicated to the case of abelian groups $G$, in which the analogues of
Alon--F{\"u}redi theorem and Schwartz--Zippel lemma for skew
polynomials are proved. Finally, Section~\ref{sec:reedmuller} is
devoted to the construction and analysis of Reed--Muller codes in
$\LL[G]$ and their connection to the Hamming setting.

\section{Preliminaries}
\label{sec:preliminaries}

\subsection{Notation}

Given a field $\K$, the elements of $\K^n$ are represented as {\bf row
  vectors} and denoted using bold face lower case letters:
$\bfa, \bfb, \dots$. However, there might be an exception to this rule:
given a finite extension $\LL$ of $\K$, a vector in $\LL^n$ whose
entries form a $\K$--basis of $\LL$ will be denoted with calligraphic
letters such as $\B$.  Matrices are denoted with capital letters:
$A, B,$ etc.  The space of matrices with $m$ rows and $n$ columns
with entries in $\K$ is denoted by $\K^{m \times n}$.  The
transposition of a vector $\bfv \in \K^n$ or a matrix
$M\in \K^{m \times n}$ is denoted by $\bfv^\top$ and $M^\top$
respectively.

Given vector spaces $V_1, V_2$ over a field $\K$ with respective bases
$\B_1, \B_2$ and a $\K$--linear map $f : V_1 \rightarrow V_2$, we
denote by $A(f, \B_1, \B_2)$ the matrix representation of $f$ in these
bases. That is to say, $A(f, \B_1, \B_2)$ is the matrix whose {\bf
  columns} are the decompositions in $\B_2$ the elements $f(\bfb)$
when $\bfb$ ranges over the basis $\B_1$.  Given a vector $x \in V_1$,
we denote $\bfx \in \K^{\dim V_1}$ its representation in basis
$\B_1$. Then, the vector $\bfy \in \K^{\dim V_2}$ such that
\[\bfy^\top = A(f, \B_1, \B_2)\cdot \bfx^\top\] is
the representation of $f(x)$ in the basis $\B_2$.
Finally, when $\B = \B_1 = \B_2$, the matrix is denoted by $A(f, \B)$.

According to this definition, the {\em kernel} of a matrix is referred to its
{\bf right} kernel, i.e. given $M \in \K^{m \times n}$
\[
  \ker M \mydef \{ \bfx \in \K^n ~|~ M\cdot \bfx^\top = 0 \}.
\]

\subsection{Skew group algebras}
\label{subsec:skew-group-algebra}

Let $\LL/\K$ be a Galois extension of
finite degree $N \mydef [\LL:\K]$, and
$G \mydef \Gal(\LL/\K)=\{g_1,\ldots, g_N\}$ be its Galois group. The
group algebra $\LL[G]$ is defined as
$$\LL[G] \mydef \left\{ \sum_{i=1}^N a_i g_i \mid a_i \in \LL \right\}.$$
The set $\LL[G]$ is naturally an $\LL$-vector space of dimension
$N$. It also has a ring structure via the multiplication $*$ defined
on monomials by $(a_ig_i) * (a_jg_j) =(a_ia_j)(g_ig_j)$ and then
extended by associativity and distributivity. However, in this paper
we will not consider this ring structure, but the one defined by the
composition $\circ$, that is given on monomials by
$$(a_ig_i)\circ (a_jg_j)=(a_i g_i(a_j))(g_ig_j),$$
and then extended by associativity and distributivity. With this
operation $\LL[G]$ is a non-commutative ring.  In addition, every
element $a = \sum_{i}a_i g_i \in \LL[G]$ can be seen as a $\K$-linear
map
\begin{equation}\label{eq:group_algebra_to_endomorphism}
\left\{\begin{array}{rcl}
  \LL & \longrightarrow & \LL \\
  x & \longmapsto & a(x) \mydef \sum_i a_i g_i(x).
\end{array}\right.
\end{equation}
\begin{theorem}
  The map sending every $a \in \LL [G]$ onto the corresponding $\K$--endomorphism of $\LL$ is a $\K$-linear isomorphism between
  $\LL[G]$ and $\End_{\K}(\LL)$.
\end{theorem}
\begin{proof}
The map is clearly $\K$-linear. Moreover, $G=\{g_1,\ldots,g_N\}$ is a set of distinct characters $\LL^\times \rightarrow \LL^\times$, defined as $x \mapsto g_i(x)$. Hence, by Artin's  theorem of independence of characters the map is injective.  The claim follows then by
   observing that both $\LL[G]$ and $\End_{\K}(\LL)$ have dimension  $N^2$ over $\K$. 
\end{proof}

\subsection{Trace of extension fields and its duality theory}\label{subsec:Trace}

For a Galois extension $\LL/\K$, the \emph{trace
  map} is a special element in $\LL[G]$ which gives rise to a
well-known duality theory.

\begin{definition}
  Let $G=\Gal(\LL/\K)$ be the Galois group of the extension
  $\LL/\K$. Then, the \emph{trace map} is defined as
  \[
    \Tr:
\left\{  \begin{array}{rcl}
     \LL &  \longrightarrow  & \K\\
    x & \longmapsto & \sum_{g \in G} g(x).
  \end{array}\right.
  \]
  The corresponding element of $\LL[G]$ is
  $\mathrm{Tr} \mydef \sum_{g \in G} g$.
\end{definition}

It is well-known that for separable extensions, and hence for Galois
extensions, the trace map induces a duality between $\LL$ and
$\Hom_{\K}(\LL,\K)$.

\begin{theorem}[Duality of the trace]
  \label{thm:dualtrace}
  Let $\LL/\K$ be a Galois extension. The map
  \begin{equation}\label{eq:trace_bilinear}
    \langle \cdot, \cdot \rangle_{\mathrm{tr}}:
    \left\{
      \begin{array}{rcl}  \LL \times \LL & \longrightarrow & \K \\
        (x,y) & \longmapsto & \Tr(xy)
      \end{array}
    \right.
  \end{equation}
    is a symmetric nondegenerate bilinear form, which induces a duality
    isomorphism 
    \[
      \left\{
        \begin{array}{rcl}
          \LL & \longrightarrow & \Hom_\K(\LL,\K) \\
          x & \longmapsto & T_x
        \end{array}
      \right.
    \]
  where $T_x(y)=\Tr(xy)$ for every $y \in \LL$.
\end{theorem}



The duality result in Theorem \ref{thm:dualtrace} also implies that
for any ordered basis $\B=(\bb_1,\ldots,\bb_N)$ of $\LL/\K$ there exists a
\emph{dual (ordered) basis} $\B^*=(\bb_1^*,\ldots, \bb_N^*)$ with respect
to the bilinear form $\langle \cdot, \cdot \rangle_{\mathrm{tr}}$. Such a dual basis satisfies
\begin{equation}\label{eq:dualbasis}
  \Tr(\bb_i\bb_j^*)=\begin{cases} 1 & \mbox{ if } i=j \\
  0 & \mbox{ if } i \neq j. \end{cases}\end{equation}

\subsection{Adjunction}\label{subsec:adjunction}
The trace bilinear form $\langle \cdot, \cdot \rangle_{\mathrm{tr}}$
introduced in Theorem~\ref{thm:dualtrace}
Equation~\eqref{eq:trace_bilinear} yields a notion of adjunction.
Given $f\in \LL[G]$, the {\em adjoint} of $f$ with respect to the
trace bilinear form $\langle \cdot, \cdot \rangle_{\mathrm{tr}}$ is
denoted by $\tau(f)$. It is the unique element $\tau(f) \in \LL[G]$
satisfying
\begin{equation}\label{eq:def_adjunct}
    \forall x,y \in \LL,\quad \langle f(x), y \rangle_{\mathrm{tr}} = 
    \Tr(f(x)y) = \langle x, \tau(f)(y) \rangle_{\mathrm{tr}}.
\end{equation}

\begin{lemma}
    The adjunction map $\tau : \LL[G] \rightarrow \LL[G]$ is a $\K$--linear map
    satisfying
    \begin{enumerate}[(i)]
        \item\label{it:L_sym} $\forall a \in \LL,\ \tau(a) = a$;
        \item\label{it:G_orth} $\forall g \in G,\ \tau(g) = g^{-1}$;
        \item\label{it:contrav} $\forall u, v \in \LL[G],\ \tau(u\circ v) =
            \tau(v) \circ \tau (u)$;
        \item\label{it:involution} $\tau$ is an involution, i.e. $\forall u \in \LL[G],\ 
        \tau \circ \tau (u) = u$.
    \end{enumerate}
\end{lemma}

\begin{proof}
  For any $a, x, y \in \LL$, we have
  $\langle ax, y \rangle_{\mathrm{tr}} = \Tr (axy) = \Tr (xay) =
  \langle x, ay \rangle_{\mathrm{tr}}$, which proves (\ref{it:L_sym}).
  Let $g \in G$ and $x,y \in \LL$, we have
  $\Tr (g(x)y) = \Tr(g(xg^{-1}(y))) = \Tr(xg^{-1}(y))$.  This proves
  (\ref{it:G_orth}).  Finally (\ref{it:contrav}) is a direct
  consequence of (\ref{eq:def_adjunct}) and
  (\ref{it:involution}) is a consequence of the symmetry of
  $\langle \cdot, \cdot \rangle_{\mathrm{tr}}$.
\end{proof}

As a consequence, we get an explicit definition of $\tau$:
\begin{equation}\label{eq:explicit_tau}
  \tau : \left\{
\begin{array}{clc}
  \LL[G] & \longrightarrow & \LL[G]\\
  u= \sum_{g \in G} u_g g & \longmapsto & \sum_{g \in G} g(u_{g^{-1}}) g.
\end{array}
\right.
\end{equation}

Actually, $\tau$ can be seen as a \enquote{transpose} map in $\LL[G]$.
In particular, if there exists an orthogonal $\K$--basis $\B$ of $\LL$
with respect to $\langle \cdot, \cdot \rangle_{\mathrm{tr}}$, then for
any $c \in \LL [G]$ we have $A(\tau(c), \B) = A(c, \B)^\top$.

Observe that this notion is well-known and studied in the context of
finite fields (see \cite{sheekey2016new, lunardon2018generalized}),
which we illustrate in the following example.
  
  \begin{example}
    Suppose that $\K=\Fq$ and $\LL=\F_{q^N}$. We have that
    $G=\Gal(\F_{q^N}/\Fq)=\langle \theta \rangle$, where $\theta$ is
    the $q$-Frobenius automorphism. Then all the elements of the
    Galois group are of the form $\theta^i(\aalpha)=\aalpha^{q^{i}}$,
    for $\aalpha \in \F_{q^N}.$ Now, fix an element
    $a\in\F_{q^N}[\theta]$ that we write as
    $a=\sum_{i=0}^{N-1}a_i\theta^i$. Hence, the adjoint of $a$ is
  \[
    \tau(a)=\sum_{i=0}^{N-1}\theta^i(a_{N-i})\theta^i=\sum_{i=0}^{N-1}a_{N-i}^{q^i}\theta^i,\]
  where by convention, $a_N \mydef a_0$.  It is clear that this
  coincides with the usual notion given for example in
  \cite{sheekey2016new}.
  \end{example}

\section{Rank metric and Moore matrices over arbitrary Galois
  extensions}
\label{sec:rank}
In this section we focus on the elements of $\LL[G]$,
where $G$ is the Galois group of an arbitrary Galois extension $\LL/\K$.
In particular, we show that we can determine the rank of any
element in several equivalent ways.

\begin{definition}
  Let $\LL/\K$ be a field extension, and let $M$ be a positive
  integer. For a given vector $\bfv=(\vvv_1,\ldots, \vvv_M) \in \LL^M$, we
  define the \emph{$\K$-rank of $\bfv$}, as the quantity
  \[
  \rk_{\K}(\bfv) \mydef \dim_{\K} \Span{\K}{\vvv_1,\ldots, \vvv_M}.
  \]
\end{definition}

We now introduce the analogue of the Moore/Wronskian matrix, for any
finite Galois group $G$.

\begin{definition}
  Let $G=\Gal(\LL/\K)=\{g_1,\ldots, g_N\}$ and $\bfv \in \LL^N$. We
  define the \emph{$G$-Moore matrix} of $\bfv$ as
  \[
  M_G(\bfv) \mydef \begin{pmatrix}g_1(v_1) & g_1(v_2) & \cdots & g_1(v_N) \\
    g_2(v_1) & g_2(v_2) & \cdots & g_2(v_N) \\
    \vdots & \vdots & & \vdots \\
    g_N(v_1) & g_N(v_2) & \cdots & g_N(v_N) \\
  \end{pmatrix} \in \LL^{N \times N}.
  \]
\end{definition}

Given an ordered $\K$--basis $\mathcal B$ of $\LL$, one can define in
a very similar fashion the Moore matrix $M_G (\mathcal B)$. In 
addition, this matrix is related to the Moore matrix of the dual basis $\mathcal B^*$
defined in Section~\ref{subsec:Trace}.

\begin{lemma}\label{lem:inverseMoore}
  Let $\B = (\beta_1, \dots, \beta_N)$ be an ordered basis of $\LL/\K$. Then
  \[
  M_G(\B)^{-1} = M_G(\B^*)^\top,
  \]
  where $\B^*$ is the dual basis of
  $\B$ with respect to the bilinear form $\langle \cdot, \cdot \rangle_{\mathrm{tr}}$.
\end{lemma}

\begin{proof}
  The $(i,j)$-th entry of $M_G(\B^*)^\top M_G(\B)$ is equal to
  $\sum_{\ell} g_\ell(\bbeta_i^*)g_{\ell}(\bbeta_j)=\Tr(\bbeta_i^*\bbeta_j)$. Therefore,
  by \eqref{eq:dualbasis}, we get
  $M_G(\B^*)^\top M_G(\B)=\mathrm{Id}$.
\end{proof}

A strong interest of the Moore matrix lies in the next statement.

\begin{proposition}
  \label{prop:rkMoore=rk}
  For every $\bfv \in \LL^N$, it holds
  \[
  \rk_\LL(M_G(\bfv)) = \rk_{\K}(\bfv).
  \]
\end{proposition}

\begin{proof}
  Set
  $r \mydef \rk_\K(\bfv) = \dim_{\K}\Span{\K}{
    \vvv_1,\ldots,\vvv_N}$. We want to prove that
  $\rk_\LL(M_G(\bfv))=r$.  By definition of $\rk_\K(\bfv)$, there exist
  an $r$--tuple of $\K$--linearly independent elements
  $\uuu_1,\ldots, \uuu_r \in \LL$ and an invertible matrix $S \in \K^{N\times N}$, such that
  $\Span{\K}{\vvv_1,\ldots, \vvv_N}= \Span{\K}{\uuu_1,\ldots, \uuu_r}$
  and
  \[
  \bfv \cdot S=(\uuu_1,\ldots,\uuu_r,0,\ldots,0)=:\bfu.
  \]
  Observe that $M_G(\bfv)\cdot S = M_G(\bfv \cdot S) =M_G(\bfu)$ since
  $S$ is defined over $\K$ and hence fixed by $G$. Consequently, the
  last $N-r$ columns of $M_G(\bfv)\cdot S$ are zero. Therefore
  \[
    \rk_{\LL}(M_G(\bfv)) = \rk_{\LL}(M_G(\bfv))\cdot S =  \rk_{\LL}(M_G(\bfu)) \leq r.
  \]
  Now, let us prove that the $r$ first columns of $M_G(\bfu) =
  M_G(\bfv)\cdot S$ are $\LL$--linearly independent. Suppose that there
  exist $\lambda_1, \dots, \lambda_r \in \LL$ satisfying
  \begin{equation}\label{eq:linear_relation_columns}
    \forall i \in \{1, \dots, N\}, \quad \sum_{j=1}^r \lambda_j g_i(u_j) = 0.
  \end{equation}
  Without loss of generality, one can suppose that $\lambda_1 \neq 0$.
  By Theorem~\ref{thm:dualtrace}, there exists $a \in \LL$ such that
  $\Tr(a\lambda_1) \neq 0$. Thus, after possibly replacing
  $\lambda_1, \dots, \lambda_r$ by $a\lambda_1, \dots, a\lambda_r$,
  one can assume that there exist $\lambda_i$'s $\in \LL$ satisfying
  (\ref{eq:linear_relation_columns}) and such that
  $\Tr(\lambda_1) \neq 0$.
  Next, (\ref{eq:linear_relation_columns}) is equivalent to
  \[
    \forall i \in \{1, \dots, N\}, \quad \sum_{j=1}^r g_i^{-1}(\lambda_j)u_j = 0.
  \]
  Summing up these $N$ equations, we get a $\K$--linear relation on
  the $u_i$'s:
  \[
    \Tr(\lambda_1)u_1 + \cdots + \Tr (\lambda_r) u_r = 0
  \]
  and this linear relation is nontrivial since
  $\Tr(\lambda_1) \neq 0$.  This yields a contradiction since the
  $u_i$'s are $\K$--linearly independent. Therefore:
  \[
    r = \rk_\K(\bfv) = \rk_\LL(M_G(\bfu)) = \rk_\LL(M_G(\bfv)).
  \]
\end{proof}

As a consequence, we get a generalization of the well-known result
over finite fields that characterizes bases of extension fields in
terms of their associated Moore matrix.

\begin{corollary}
  A vector $\bfv \in \LL^N$ is an ordered basis of $\LL/\K$ if and only
  if $\det(M_G(\bfv))\neq 0$.
\end{corollary}

The previous results give properties of the rank metric
on $\LL^N$ by relating the rank of an element with the rank
of its Moore matrix. Let us now investigate the rank metric
in $\LL [G]$.

\begin{definition}
  Let $\LL/\K$ be a finite extension with Galois group $G$. The
  $\K$-rank of an element $a \in \LL[G]$ is defined as the rank of the
  corresponding $\K$--endomorphism of $\LL$ (see
  \eqref{eq:group_algebra_to_endomorphism}).
\end{definition}

Given a vector $\bfb=(\bb_1,\ldots, \bb_M)\in \LL^M$, let us now
define the {\em evaluation map}
\begin{equation}\label{eq:ev_map}
  \mathrm{ev}_\bfb:
  \left\{
    \begin{array}{rcl}
   \LL[G] & \longrightarrow &  \LL^M\\
  a & \longmapsto & (a(\bb_1),\ldots, a(\bb_M)).
    \end{array}
  \right.
\end{equation}
For $a \in \LL[G]$, the vector $\mathrm{ev}_\bfb(a) \in \LL^M$ is
called the \emph{evaluation vector} of $a$ at $\bfb$. One
can easily see that $\rk_{\K}(a) = \rk_\K(\ev_{\mathcal{B}}(a))$ for
every basis $\mathcal{B}$ of $\LL/\K$.

\begin{definition}
The \emph{left-annihilator} of an element $a \in \LL[G]$ is defined as
\[
\Ann_{\LL[G]}(a) \mydef \left\{ f \in \LL[G] \mid f \circ a = 0 \right\}.
\]  
\end{definition}

Observe that $\Ann_{\LL[G]}(a)$ is an $\LL$-subspace and a left-ideal
in $\LL[G]$.

\begin{proposition}
  For every $a \in \LL[G]$  we have
  \[
  \rk_\K(a) = \dim_{\LL}\left(\LL[G]/\Ann_{\LL[G]}(a)\right)\,.
  \]
\end{proposition}

\begin{proof}
  Let us set $\bfv = \ev_{\mathcal{B}}(a)$ for some basis
  $\mathcal{B} = (\beta_1,\dots, \beta_N)$ of $\LL/\K$.  We have
  \[
    \dim_{\LL}\left(\LL[G]/\Ann_{\LL[G]}(a)\right) =
    N-\dim_{\LL}(\Ann_{\LL[G]}(a)),
  \]
  and Proposition~\ref{prop:rkMoore=rk} yields
  $\rk_\K(a) = \rk_\K(\bfv) = \rk_\LL(M_G(\bfv))$.  To prove the result, we will
  prove that $\Ann_{\LL [G]}(a)$ and $\ker_\LL M_G(\bfv)^\top$ are
  isomorphic.  Indeed, consider the natural $\LL$-isomorphism
  \[\varphi:
    \left\{
  \begin{array}{rcl}  \LL[G] & \longrightarrow & \LL^{N} \\
    \sum_i\llambda_ig_i & \longmapsto & (\llambda_1,\ldots, \llambda_N).
  \end{array}
  \right.
  \]
  One can see that
  $\varphi(\Ann_{\LL[G]}(a)) = \ker_\LL(M_G(\bfv)^\top)$.  Indeed,
  $(\llambda_1,\ldots, \llambda_N) \in \ker_\LL(M_G(\bfv)^\top)$ if
  and only if
  \[
    \forall j \in \{1,\ldots, N\}, \quad 0=\left(\sum_i
      \llambda_ig_i\right)(v_j) = \left(\sum_i
      \llambda_ig_i\right)(a(\bbeta_j)) = \left(\left(\sum
        \llambda_ig_i\right) \circ a \right)(\bbeta_j).
  \]
  Since
    $\B$ is a basis, this holds if and only if
    $(\sum \llambda_ig_i) \circ a =
    0$, which is equivalent to say that
    $\varphi^{-1}(\llambda_1,\ldots,\llambda_N)=(\sum \llambda_ig_i)
    \in \Ann_{\LL[G]}(a)$. This proves that
    $\varphi(\Ann_{\LL[G]}(a)) \supseteq
    \ker_\LL(M_G(\bfv))$ and the converse inclusion can be proved in a
    similar fashion.
\end{proof}

To sum up, let $a \in \LL[G]$ and define
$\w_I(a) \mydef \dim_\LL(\LL[G]/\Ann_{\LL[G]}(a))$. If we set
$\bfv = \ev_{\mathcal{B}}(a)$ for some basis $\mathcal{B}$ of
$\LL/\K$, then we have proved:
\[
\rk_\K(a) = \rk_\LL(M_G(\bfv)) = \w_I(a)\,.
\]

\section{Dickson matrices for elements in \texorpdfstring{$\LL[G]$}{L[G]}}
\label{sec:Dickson}

In this section, we study {\em Dickson matrices} in the context of arbitrary
Galois groups. 
Before giving their definition let us introduce a notation. Consider 
 the left action of $G = \Gal(\LL/\K)$
on itself and denote by
$\sigma_i \in \mathfrak S_N$  the permutation associated to
$g_i \in G$, \emph{i.e.}  $g_ig_j=g_{\sigma_i(j)}$ for all
$i,j \in \{1,\dots,N\}$.

\begin{definition}\label{def:Dickson_mat}
  Let us fix some ordering $(g_1, \ldots , g_N)$ of the group $G$
  Let $a = \sum_ {i} a_i g_i \in \LL [G]$. The {\em $G$--Dickson matrix}
  associated to $a$ is defined as
  $D_G(a) = (d_{i,j}) \in \LL^{N \times N}$ defined by
 \[
 d_{i,j} = g_j(a_{\sigma_j^{-1}(i)}), \quad\quad \forall i,j \in
 \{1, \dots, N\},
 \]
\end{definition}

\begin{example}
  When $\K=\Fq$ and $\LL=\F_{q^N}$, we have that
  $G=\Gal(\F_{q^N}/\Fq)=\langle \theta \rangle$, where
  $\theta$ is the $q$-Frobenius automorphism.
  Then choosing the ordered $\F_{q^N}$--basis $(\Id, \theta, \dots, \theta^{N-1})$
  of $\F_{q^N}[G]$, we get that the $G$--Dickson matrix $D_G(a)$ of an element $a = a_1 \textrm{Id} + a_2 \theta + \cdots + a_N \theta^{N-1} \in \F_{q^N}[G]$ is given by:
  \[
  D_G(a)=\begin{pmatrix} a_1 & a_N^q & \cdots & a_2^{q^{N-1}} \\
    a_2 & a_1^q & \cdots & a_3^{q{N-1}} \\
    \vdots &  & \ddots &  \\
    a_N & a_{N-1}^q & \cdots & a_1^{q^{N-1}} \end{pmatrix}.
  \]
  This matrix is usually known as the
  Dickson matrix associated to
  $a(x)=\sum_{i=1}^N a_i x^{q^{i-1}} \in \mathcal{L}[x]$, where we recall that $\mathcal L [x]$
  denotes the ring of linearized polynomials. Since
  $\mathcal{L}[x]/(x^{q^N}-x) \cong \F_{q^N}[G]$, this explains the
  relation between the $G$--Dickson matrix and the usual Dickson
  matrix over finite fields.
\end{example}

In the sequel we give two distinct interpretations of these matrices.

\subsection{The right multiplication map}\label{subsec:multiplication_map}
If $a = \sum_{i} a_i g_i \in \LL[G]$, then for
every $j \in \{1,\ldots, N\}$ we have
$$ g_j \circ \left(\sum_{i=1}^N a_i g_i \right) =
\sum_{i=1}^Ng_j(a_i)g_jg_i=\sum_{i=1}^Ng_j(a_i)g_{\sigma_j(i)}.$$
Now, let us consider the $\LL$-linear map
\begin{equation}\label{eq:delta_a}
  \mu:
  \map{\LL[G]}{\textrm{Hom}_\LL (\LL[G], \LL[G])}{a}{(f \mapsto f \circ a).}
\end{equation}

\begin{proposition}
  Let $a \in \LL [G]$. Then, the matrix representing $\mu(a)$
  in the basis $(g_1, \dots, g_N)$ is the $G$--Dickson matrix $D_G(a)$.
\end{proposition}

\begin{remark}
  The $G$-Dickson matrix is the matrix associated to the $\LL$-linear
  map $\mu(a)$, given by the right composition by $a$. One can also
  consider the map $\LL[G] \to \LL[G]$ given by the left
  composition $f \mapsto a \circ f$. However this map is only
  semilinear.
\end{remark}

\subsection{The element of the group algebra after a base field
  extension}\label{subsec:base_extension}
Given $a \in \LL[G]$, the element $a$ induces a
$\K$--endomorphism $a : \LL \rightarrow \LL$. We claim that the
transposition of its $G$--Dickson matrix represents this endomorphism after
a base field extension. To understand this fact, we introduce the map
\[
  \nu : \map{\LL[G]}{\textrm{Hom}_\LL (\LL \otimes_\K \LL, \LL \otimes_\K \LL)}{
  a}{\Id \otimes a}
\]
and will study in depth the maps of the form
$
  \Id \otimes a : \LL \otimes_\K \LL \rightarrow \LL \otimes_\K \LL.
$
Let $\alpha$ be a primitive element of $\LL/\K$, and consider the
$\K$--linear map given by the multiplication by $\alpha$
\[
  m_\alpha :\left\{
  \begin{array}{ccc}
    \LL & \longrightarrow & \LL \\
    x & \longmapsto & \alpha x.
  \end{array}
  \right.
\]
In the $\K$--basis $(1, \alpha, \alpha^2, \dots, \alpha^{N-1})$ this map is
represented by the companion matrix of the minimal polynomial of $\alpha$
over $\K$.  Therefore, its eigenvalues are nothing but
the $g(\alpha)$ for $g \in G$. For a suitable choice of
eigenvectors basis in $\LL \otimes_\K \LL \simeq \LL^N$ the map
$\Id \otimes m_\alpha$ has a diagonal matrix representation
\begin{equation}\label{eq:diag_ma}
  \begin{pmatrix}
    \alpha & & & (0) \\
      & g_2(\alpha) &  \\
      & & \ddots &  \\
     (0) & &        &  g_N(\alpha)
  \end{pmatrix},
\end{equation}
where we ordered the elements of $G$ so that $g_1 = \Id$.  This matrix
representation is associated to a basis of eigenvectors of
$\LL \otimes_\K \LL$. Let us make a particular choice of normalisation
for them.  Choose $v \in \LL \otimes_\K \LL$ an eigenvector of
$\Id \otimes m_\alpha$ with respect to the eigenvalue $\alpha$. That is
to say $(\Id \otimes m_\alpha)(v) = \alpha\cdot v.$
For $g \in G$ we define  $v_g \mydef (\Id \otimes g^{-1})(v)$.

\begin{proposition}\label{prop:eigenvectors}
  Let $g \in G$. Then
  $v_g$ is an eigenvector of $\Id \otimes m_\alpha$ with respect to the
  eigenvalue $g(\alpha)$.
\end{proposition}

\begin{proof}
  First, note that $g \circ m_\alpha = m_{g(\alpha)} \circ g$. 
  Therefore, we have
  \begin{align*}
    (\Id \otimes m_\alpha) \circ(\Id \otimes g^{-1}) (v)  & =
                                 \Id \otimes (m_\alpha \circ g^{-1}) (v)\\
     & = \Id \otimes (g^{-1}\circ m_{g(\alpha)} ) (v)\\
     & = (\Id \otimes g^{-1})\circ (\Id \otimes m_{g(\alpha)}) (v).
  \end{align*}
  Since $\alpha$ is a primitive element of $\LL/\K$, there exists a
  polynomial $P \in \K[X]$ such that $g(\alpha) = P(\alpha)$ and hence
  $P(m_\alpha) = m_{g(\alpha)}$.  Moreover, since $v$ is an
  eigenvector of $\Id \otimes m_\alpha$ with respect to the eigenvalue $\alpha$, then
  it is an eigenvector of $P(\Id \otimes m_\alpha) = \Id \otimes P(m_\alpha)$
  with respect to the eigenvalue $P(\alpha)$.
  Therefore,
  \begin{align*}
    (\Id \otimes m_\alpha) \circ(\Id \otimes g^{-1}) (v)  &
                  = (\Id \otimes g^{-1})\circ
                          (\Id \otimes m_{g(\alpha)} ) (v) \\
                & = (\Id \otimes g^{-1})(P(\alpha)(v)) \\
                & =  P(\alpha)\cdot(\Id \otimes g^{-1})(v) \\
                & =  g(\alpha)\cdot(\Id \otimes g^{-1})(v).
  \end{align*}
  In summary $v_g \mydef (\Id \otimes g^{-1})(v)$ is an eigenvector
  of $\Id \otimes m_\alpha$ with respect to $g(\alpha)$.
\end{proof}

Therefore, in the basis ${(v_g)}_{g\in G}$ the multiplication by an
element $\alpha \in \LL$ is represented by a diagonal matrix. Next,
the action of elements of $G$ will be represented by permutation matrices
as suggests the next statement.

\begin{proposition}
  \label{prop:composition-eigenvector}
  Let $g , h \in G$, then $(\Id \otimes g)(v_h) = v_{hg^{-1}}$.
\end{proposition}

\begin{proof}
  $
    (\Id \otimes g)\circ (\Id \otimes h^{-1})(v) = (\Id \otimes
    {(hg^{-1})}^{-1})(v) = v_{hg^{-1}}.
  $
\end{proof}

As a conclusion, $G$ acts by permutation on eigenvectors $(v_g)_{g \in G}$.

\subsection{Relating these two approaches}
Now, let us try to relate matrix representations of $\mu(a)$
and $\nu(a)$.

\begin{theorem}
  Let \[
    \Lambda : \map{\LL [G]}{\LL \otimes_\K \LL}{\sum_g a_g g}{\sum_g a_gv_g,}
  \]
  then for any $a \in \LL [G]$, we have
  \[
    \mu(\tau(a)) = \Lambda^{-1} \circ \nu(a) \circ \Lambda,
  \]
  where $\tau$ is the adjunction map introduced in
  Section~\ref{subsec:adjunction}. From the matrix point of view:
  \[
    D_G(a) = A(\mu(a), (g_1,\dots, g_N)) = A(\nu(a), (v_{g_1},\dots,
    v_{g_N}))^\top.
  \]
\end{theorem}

\begin{proof}
  Note first that the map $\nu$ is a ring homomorphism, while $\mu$ is
  a ring anti-homomorphism: for any $a, b \in \LL [G]$, we have
  $\mu(a\circ b) = \mu(b) \circ \mu(a)$. Thus, we introduce the map
  $\mu' : a \mapsto \mu(\tau (a))$ which is a ring homomorphism and we
  will show that $\mu'$ and $\nu$ have conjugated images under
  $\Lambda$. Since we have ring homomorphisms it is sufficient
  to prove that the property is satisfied by generators, i.e.
  elements of $\LL$ and elements of $G$.

  First consider the case of an element $a \in \LL$. We already proved
  that $\nu(a)$ has a diagonal representation in the basis $(v_g)_{g\in G}$.
  On the other hand for any $g \in G$, we have
  \[
  \mu'(a)(g) = g \circ \tau(a)
             = g \circ a
             = g(a) \cdot g.
  \]
  Hence, $g$ is an eigenvector of $\mu'(a)$ with respect to the eigenvalue
  $g(a)$ and hence has the very same matrix representation.
  Formally, $\mu'(a) = \Lambda^{-1} \circ \nu(a) \circ \Lambda$.

  Next, consider an element $g \in G$. By Proposition~\ref{prop:composition-eigenvector}, for any $h \in G$ we have $\nu(g)(v_h) = v_{hg^{-1}}$.
  On the other hand,
  \[
    \mu'(g)(h)  = h \circ \tau(g) 
                = hg^{-1}
  \]
  Hence, here again, we deduce that
  $\mu'(g) = \Lambda^{-1} \circ \nu(g) \circ \Lambda$. This concludes the proof.
\end{proof}

\begin{corollary}
  Let $\B$ be a $\K$--basis of $\LL$ and
  $\B_\LL \mydef (1 \otimes b)_{b \in \B}$ the corresponding $\LL$--basis
  of $\LL \otimes_\K \LL$. Let $\bfv$ be the representation of the eigenvector
  $v \in \LL \otimes_\K \LL$ in the basis $\B_\LL$. Then, for any $a \in \LL [G]$
  \[
    D_G(a)^\top = M_G(\bfv)^\top A(a, \B) (M_G(\bfv)^\top)^{-1}.
  \]
\end{corollary}

\begin{proof}
  The matrix $D_G(a)^\top$ represents $\nu(a)$ in the basis $(v_g)_{g \in G}$.
  On the other hand $M_G(\bfv)^\top$ can be interpreted as the change of basis
  matrix from $\B_\LL$ to $(v_g)_{g \in G}$.
\end{proof}

\subsection{Properties of Dickson matrices}
The previous observations permit first to assert the following statement.

\begin{lemma}\label{lem:transpose_Dickson}
  For any $a \in \LL [G]$, we have
  \[D_G(a)^\top = D_G (\tau (a)).\]
\end{lemma}

Next, if we define the following algebra,
$\mathcal D(\LL/\K) \mydef \{ D_G(a)^\top \mid a \in \LL[G] \}
\subseteq \LL^{N \times N}$, then we get a new ring isomorphism:
\[
  \mathcal D(\LL/\K) \cong \LL[G] \cong \End_\K(\LL)\cong \K^{N \times N}.
\]
In addition, the rank of an element of $\LL [G]$ can obviously be interpreted in
terms of the rank of its $G$--Dickson matrix, as it holds for the finite field case (see e.g. \cite{menichetti1986roots,wu2013linearized,csajbok2020scalar}).

\begin{theorem}\label{thm:rankcharacterization}
  Let $a \in \LL[G]$ and
  $\bfv \mydef (a(\beta_1), \dots, a(\beta_N))$ for some basis
  $\B = (\beta_1, \dots, \beta_N)$ of $\LL/\K$. Then,
  \[
    \rk_{\K}(a) = \w_I(a) = \rk_\LL(M_G(\bfv))
    =\rk_{\LL}(D_G(a)).
  \]
\end{theorem}

\section{Rank-metric codes}
\label{sec:RMC}

The theory of rank-metric codes has been essentially always studied
in the context of extension fields with  cyclic Galois groups.
For the special case of finite fields, the reader is referred
to~\cite{sheekey2019mrd}.
In this section, we consider the case of general Galois extensions $\LL/\K$ of
finite degree $N=[\LL:\K]=|\Gal(\LL/\K)|$.

\subsection{Equivalent representations of codes}
\label{subsec:rank-metric}

According to Sections~\ref{sec:rank} and~\ref{sec:Dickson}, we can
define the rank metric in several equivalent ways. In $\LL[G]$,
the rank distance is defined as
  $$\dist(a,b) \mydef \rk_\K(a-b), \qquad
  \mbox{ for any } a, b \in \LL[G].$$



\begin{definition}
  An $\LL$-linear \emph{rank-metric code} $\C$ is an $\LL$-subspace of
  $\LL[G]$, equipped with the rank distance. The \emph{dimension} of
  $\C$ is its dimension as an $\LL$-vector space, and its
  \emph{minimum rank-distance} is the integer
$$\dist(\C)\mydef\min \left\{ \dist(a,b) \mid a,b \in \C, a \neq b \right\}.$$
An $\LL$-linear rank-metric code $\C\subseteq \LL[G]$ of $\LL$--dimension $k$
and minimum rank distance $d$ will be also called an
$[N,k,d]_{\LL[G]}$ code, where $N\mydef|G|$, or simply $[N,k]_{\LL[G]}$
code, if the minimum rank distance is not known/relevant.
\end{definition}

Rank-metric codes have been previously studied in other ambient spaces.
First, in spaces of matrices, the rank distance is defined as
\[
  \dist:
  \left\{
\begin{array}{ccc}
  \K^{N\times M} \times \K^{N \times M} & \longrightarrow & \mathbb N \\
  (A,B) & \longmapsto & \rk_{\K}(A-B).
\end{array}
\right.
\]
Codes in this setting are usually  called \emph{matrix rank-metric codes}.
Linear codes are 
$K$--dimensional $\K$-subspaces of $\K^{N\times M}$ and they
are denoted by 
$[N\times M, K]_{\K}$ codes (or $[N\times M, K, d]_{\K}$
codes if the minimum distance is known).

As in classical literature, we can also define the rank distance on
vectors over $\LL$ as
\[
  \dist:
  \left\{
    \begin{array}{ccc}
\LL^M \times \LL^M & \longrightarrow & \mathbb N \\
  (\bfu,\bfv) & \longmapsto & \rk_\K(\bfu-\bfv).
    \end{array}
    \right.
\]
Here, codes are called \emph{vector rank-metric codes}.
Linear codes in this framework are $k$-dimensional $\LL$-subspaces of
$\LL^M$ and they are denoted by $[M, k]_{\LL/\K}$ codes
(or $[M, k, d]_{\LL/\K}$ codes if the minimum distance is known).

\medskip

\subsubsection{From vector codes to matrix codes.}
In the theory of rank-metric codes there is a procedure for going from
an $[M, k, d]_{\LL/\K}$ code to an $[N\times M, Nk, d]_{\K}$ code.
Fix an ordered basis $\B$ of $\LL/\K$, and write every element of
$\LL$ in coordinates with respect to $\B$, resulting in a column
vector in $\K^N$.  In the same way, we can transform a vector
$\bfv \in \LL^M$ to a matrix in $\K^{N \times M}$, which we denote by
$\Ext_\B(\bfv)$. Hence, for an $[M, k, d]_{\LL/\K}$ code $\C$ and a
fixed ordered basis $\B$ of $\LL/\K$ we define
\[
  \Ext_\B(\C) \mydef \left\{ \Ext_\B(\bfv) \mid \bfv \in \C \right\}
  \subseteq \K^{N \times M},
\]
which is an $[N \times M, Nk, d]_{\K}$ code.

\medskip

\subsubsection{From $\LL[G]$-codes to vector codes.}
Now, we briefly explain the relation between rank-metric codes in
$\LL[G]$ and vector rank-metric codes in $\LL^N$. Let
$\C \subseteq \LL[G]$ be an $[N,k,d]_{\LL[G]}$ code and fix an ordered
basis $\B$ of $\LL/\K$. Then, we define the code
\[
\C(\B)\mydef\left\{ \evB(c) \mid c \in \C \right\}.
\]
By Theorem \ref{thm:rankcharacterization} the map $\C \mapsto \C(\B)$
is an isometry between spaces $(\LL[G], \dist)$ and $(\LL^N, \dist)$,
and hence the code $\C(\B)$ is an $[N,k,d]_{\LL/\K}$ vector
rank-metric code. Moreover, if we fix two ordered bases $\B_1$ and
$\B_2$ of $\LL/\K$, and let $X \in \K^{N \times N}$ be the change-of-basis
matrix such that $\B_1=\B_2X$, then we have
\begin{equation}\label{eq:equivevaluationcodes}
    \C(\B_1)=\C(\B_2)\cdot X=\left\{ \bfv X \mid \bfv \in \C(\B_2) \right\}.
\end{equation}
One may note that the two codes $\C(\B_1)$ and $\C(\B_2)$ are equivalent in the sense of vector rank-metric codes (see \cite{mo14} for the finite field case.).  In
particular, they are isometric with respect to the rank metric.

\medskip

\subsubsection{From $\LL[G]$-codes to matrix codes.}
Finally, if we fix two ordered bases $\B_1$ and $\B_2$ of $\LL/\K$, we
can transform the $[N,k,d]_{\LL[G]}$ code $\C$ in the
vector code $\C(\B_1)$ and then to the matrix code $\Ext_{\B_2}(\C(\B_1))$.
This last matrix code satisfies,
\[
  \Ext_{\B_2}(\C(\B_1)) = \left\{ A(c, \B_1, \B_2) ~|~ c\in \C \right\}.
\]
In the special case
in which $\B_1=\B_2=:\B$, we get the code
\[
\Ext_\B(\C(\B)) \mydef \left\{A(c,\B) \mid c \in \C\right\}.
\]

\begin{example}
  Let us fix $\K=\Q$, and $\LL$ to be the splitting field of the
  polynomial $x^3-p$, where $p$ is a prime number. This means that
  $\LL=\Q(\zeta, \sqrt[3]{p})$, where $\zeta$ is a primitive $3$rd
  root of unity satisfying $\zeta^2+\zeta+1=0$. The Galois group
  $G=\Gal(\LL/\K)$ is isomorphic to the symmetric group
  $\mathfrak S_3$ and it is generated by the automorphisms $\sigma_1$
  and $\sigma_2$, defined as
  \[
  \sigma_1 : \left\{ \begin{array}{ll} \zeta &\mapsto \zeta^2\\
    \sqrt[3]{p} &\mapsto 
    \sqrt[3]{p}
  \end{array}\right.
  \quad \text{ and } \quad
  \sigma_2 : \left\{  \begin{array}{ll} \zeta &\mapsto \zeta\\
    \sqrt[3]{p} &\mapsto 
    \zeta\sqrt[3]{p}.
  \end{array}\right.
  \]
  Consider the $[6,3]_{\LL[G]}$ rank-metric code given by
  \[
    \C\mydef\left\{a\cdot \Id +b\cdot\sigma_1+c\cdot\sigma_2 \mid a,b,c\in\LL \right\}.
    \]
    We fix the following ordered basis
    $\B=\left(1,\zeta, \sqrt[3]{p}, \zeta\sqrt[3]{p}, \sqrt[3]{p}^2,
      \zeta\sqrt[3]{p}^2\right)$ of $\LL/\K$. Then, the
    $[6,3]_{\LL/\K}$ code $\C(\B)$ is generated by the matrix
  \[
    \begin{pmatrix} 1 & \zeta & \sqrt[3]{p} &  \zeta\sqrt[3]{p} & \sqrt[3]{p}^2 & \zeta\sqrt[3]{p}^2 \\
      1 & -(\zeta+1) & \sqrt[3]{p} & -(\zeta+1)\sqrt[3]{p}  & \sqrt[3]{p}^2  & -(\zeta+1)\sqrt[3]{p}^2  \\
      1 & \zeta & \zeta \sqrt[3]{p} & -(\zeta+1)\sqrt[3]{p}& -(\zeta+1)\sqrt[3]{p}^2 & \sqrt[3]{p}^2  \\
  \end{pmatrix}.
  \]
  Moreover, we can also determine the $[6\times 6, 18]_{\K}$ matrix code $\Ext_\B(\C(\B))$. The matrices that represent the scalar multiplication by the six elements of the basis are of the form 
  $A^iB^j$ for $i\in\{0,1\}$ and $j \in \{0,1,2\}$, where
  \[
  A =
  \begin{pmatrix}
    0&-1&0&0&0&0 \\
    1&-1&0&0&0&0 \\
    0&0&0&-1&0&0 \\
    0&0&1&-1&0&0 \\
    0&0&0&0&0&-1 \\
    0&0&0&0&1&-1
  \end{pmatrix}
  \quad \text { and } \quad B=
  \begin{pmatrix}
    0&0&0&0&p&0 \\
    0&0&0&0&0&p \\
    1&0&0&0&0&0 \\
    0&1&0&0&0&0 \\
    0&0&1&0&0&0 \\
    0&0&0&1&0&0
  \end{pmatrix}.
  \]
  This is due to the fact that $A$ and $B$ represent the
  multiplication by $\zeta$ and $\sqrt[3]{p}$ respectively. Hence, by
  writing the three row vectors of the generator matrix of $\C(\B)$
  with respect to the basis $\B$, we see that the code
  $\Ext_\B(\C(\B))$ is the $\Q$-span of the set
  $$\left\{A^iB^j, A^iB^j X, A^i B^j Y \mid 0\leq i \leq 1, 0\leq j \leq
    2\right\},$$
  where
  \[
  X = \begin{pmatrix}
    1&-1&0&0&0&0\\
    0&-1&0&0&0&0\\
    0&0&1&-1&0&0\\
    0&0&0&-1&0&0\\
    0&0&0&0&1&-1\\
    0&0&0&0&0&-1
  \end{pmatrix} \quad \text { and } \quad
  Y = \begin{pmatrix}
    1&0&0&0&0&0\\
    0&1&0&0&0&0\\
    0&0&0&-1&0&0\\
    0&0&1&-1&0&0\\
    0&0&0&0&-1&1\\
    0&0&0&0&-1&0
  \end{pmatrix}
  \]
  are the matrices representing $\sigma_1$ and
  $\sigma_2$ in the basis $\B$. In other words, $X$ and
  $Y$ are the $\Ext_{\B}$ of the vectors
  \[(1,
  -(\zeta+1), \sqrt[3]{p}, -(\zeta+1)\sqrt[3]{p}, \sqrt[3]{p}^2,
  -(\zeta+1)\sqrt[3]{p}^2) \quad {\rm and} \quad ( 1, \zeta, \zeta \sqrt[3]{p},
  -(\zeta+1)\sqrt[3]{p}, -(\zeta+1)\sqrt[3]{p}^2,
  \sqrt[3]{p}^2)\]
respectively. 
\end{example}

\subsection{Duality for rank-metric codes}

Here, we study the different notions of duality for rank-metric
codes, according to the three representations mentioned above
and how they are related. 

\medskip

\subsubsection{Matrix codes.} First, on the space of matrices we
consider the standard bilinear form for matrices, given by
\[
  \left\{
    \begin{array}{ccc}
 \K^{N\times M} \times \K^{N\times M} & \longrightarrow & \K \\
 (A,B) & \longmapsto & \mathrm{Tr}(AB^\top),
    \end{array}
  \right.
\]
where $\mathrm{Tr}$ denotes the matrix trace.
The \emph{dual code} of an $[N\times M, K]_{\K}$ code $\C$ is then
$$\C^{\perp}\mydef\big\{ A \in \K^{N \times M} \mid \mathrm{Tr}(AB^\top)=0 \mbox{ for all } B \in\C \big\}.$$
Since the standard bilinear form is nondegenerate, then $\C^\perp$ is
an $[N\times M, NM-K]$ code.

\medskip

\subsubsection{Vector codes.} For vector rank-metric codes, the
duality is always taken with respect to the standard inner
product. Hence, for an $[M,k]_{\LL/\K}$ code, its \emph{dual code} is
the $[M,M-k]_{\LL/\K}$ code given by
\[
\C^\perp\mydef\big\{ \bfu \in \LL^M \mid \bfu \cdot \bfv^\top=0 \mbox{ for all } \bfv \in \C \big\}.
\]

\medskip

\subsubsection{$\LL[G]$-codes.} Finally, we introduce the following
bilinear form on $\LL[G]$ --- which is also called standard bilinear
form over finite fields --- defined as
\[
  \langle\cdot  , \cdot \rangle_{\LL[G]}:
  \left\{
\begin{array}{ccc}
   \LL[G] \times \LL[G] & \longrightarrow & \LL\\
   \left(a = \sum\limits_{g \in G} a_g g,\ b = \sum\limits_{g\in G} b_g g \right) & \longmapsto & \sum\limits_{g \in G} a_g b_g.
\end{array}
\right.
\]
This bilinear form is also nondegenerate, and given an
$[N,k]_{\LL[G]}$ rank-metric code, we define its \emph{dual code} as
the $[N,N-k]_{\LL[G]}$ code
\[
  \C^\perp\mydef\left\{ a \in \LL[G] \mid \langle a, b\rangle_{\LL[G]}=0 \mbox{
      for all } b \in \C\right\}.
\]

In Section~\ref{subsec:rank-metric}, we have seen how codes in these
three points of view are related. This can be extended to duality. For
instance, the relation between the duality of matrix and vector
rank-metric codes over finite fields has been already studied in
\cite{grant2008duality, ravagnani2016rank}. With the same proof, it is
easy to see that for any $[M,k]_{\LL/\K}$ code $\C \subseteq \LL^M$
and any ordered basis $\B$ of $\LL/\K$ with dual basis $\B^*$, it
holds
\[
\Ext_\B(\C)^{\perp}=\Ext_{\B^*}(\C^\perp).
\]

Now, let us show how rank-metric codes in $\LL[G]$ are related to vector
rank-metric codes in $\LL^N$. Let $\aalpha \in \LL$ be a normal element
of $\LL/\K$, \emph{i.e.} the set $\{ g(\aalpha) \mid g \in G \}$ is a
basis of $\LL/\K$. If we fix some ordering for the elements of $G$,
say $g_1, \dots, g_N$, then we get an ordered normal basis
$\bfa = (g_1(\aalpha), \dots, g_N(\aalpha)) \in \LL^N$. One can prove
that the dual basis of an ordered normal basis is normal with respect
to the same ordering of elements of $G$.

\begin{theorem}\label{thm:dualnormalbasis}
  Let $\bfa = (g_1(\aalpha), \dots, g_N(\aalpha)) \in \LL^N$ be an
  ordered normal basis, where $\aalpha \in \LL$. Then, there exists
  $\bbeta \in \LL$ such that
  $\bfb = (g_1(\bbeta), \dots, g_N(\bbeta)) \in \LL^N$ is the dual basis
  of $\bfa$.
\end{theorem}

\begin{proof}
  Without loss of generality, assume that $g_1$ is the identity
  element. Let $\bfb = (\bb_1, \dots, \bb_N) \in \LL^N$ be the unique dual
  basis of $\bfa$ and define $\beta:=b_1$. Then, by $G$-invariance of the trace, 
  \[
    \Tr(g_i(\bbeta) g_j(\aalpha)) = 
    \Tr(\bb_1 g_i^{-1}g_j(\aalpha))
  \]
  is $0$ if $i \ne j$, and $1$ otherwise. Hence,
  $(g_1(\bbeta), \dots, g_N(\bbeta))$ is dual to $\bfa$, and by uniqueness,
  $\bb_j = g_j(\bbeta)$ for every $j$.
\end{proof}

From this result, we can relate the notions of duality of rank-metric
codes, when $G$ is abelian.

\begin{theorem}\label{thm:dualLG}
  Let $\C$ be an   $[N,k]_{\LL[G]}$ code and let $\B$ be an ordered  basis of $\LL/\K$. Moreover, assume that $G$ is abelian. Then
  \[
  \Ext_\B(\C^\perp) = \Ext_{\B^*}(\C)^\perp.
  \]
\end{theorem}

\begin{proof}
  First, we fix a normal basis
  $\bfa=(g_1(\aalpha),\ldots,g_N(\aalpha))$, which always exists
  thanks to the normal basis theorem. By Theorem
  \ref{thm:dualnormalbasis} there exists $\bbeta \in \LL$, such that
  $\bfb\mydef(g_1(\bbeta),\ldots,g_N(\bbeta))=\bfa^*$. Moreover, we
  have
  \[
  \begin{aligned}
    \ev_\bfa(g_i)  \ev_\bfb(g_j)^\top &= \sum_{\ell=1}^N g_i(a_\ell) g_j(\bb_\ell)
    = \sum_{\ell=1}^N g_i(g_\ell(\aalpha)) g_j(g_\ell(\bbeta))\\
    &= \sum_{\ell=1}^N g_\ell(g_i(\aalpha) g_j(\bbeta)) =
    \Tr(g_i(\aalpha) g_j(\bbeta)) \\ 
    &= \delta_{i,j}= \langle g_i, g_j \rangle_{\LL[G]}\,.
  \end{aligned}
  \]
  This shows that in this case $\C^\perp(\bfa)=(\C(\bfb))^\perp$.
  
  Now, suppose that $\B=(\bb_1,\ldots,\bb_N)$ is a generic ordered
  basis. There exists an invertible matrix $X \in \K^{N \times N}$ such that $\B=\bfa
  X$. Moreover, we also have that $\B^*=\bfb (X^{-1})^\top$. Hence, we
  get
  \[
  \begin{aligned}
    \ev_\B(g_i) \ev_{\B^*}(g_j)^\top &=  \ev_{\bfa X}(g_i) \ev_{\bfb (X^{-1})^\top}(g_j)^\top \\
    &= \ev_\bfa(g_i)X\left(\ev_{\bfb}(g_j)(X^{-1})^\top \right)^\top \\
    &=\ev_\bfa(g_i)X X^{-1} \ev_{\bfb}(g_j)^\top \\
    &= \delta_{i,j}= \langle g_i, g_j \rangle_{\LL[G]}\,.
  \end{aligned}
  \]
\end{proof}

\section{Error-correcting pairs in \texorpdfstring{$\LL[G]$}{L[G]}}
\label{sec:ECP}

In this section, we make a first step towards decoding codes seen as
$\LL$-subspaces of $\LL[G]$. We adapt the notion of rank
error-correcting pairs (rank-ECP) introduced by
Mart{\'{\i}}nez{-}Pe{\~{n}}as and Pellikaan~\cite{Martinez-PenasP17},
which themselves were counterparts of Hamming metric error-correcting
pairs~\cite{Pellikaan92}.

\medskip

\noindent {\bf Note.} From now on and for convenience sake, we always suppose
that the group $G$ is equipped with some total ordering and we allow
ourselves to index rows and columns of matrices with elements of $G$.
Given $A \in \LL^{|G|\times |G|}$ and $g,h \in G$, we denote by $A_{g,h}$
the entry of $A$ at row $i$ and column $j$, where $g$ (resp. $h$) is the
$i$--th (resp. $j$--th) element of $G$ with respect to this ordering.
As a consequence, from Definition~\ref{def:Dickson_mat},
$G$--Dickson matrices are defined as
  \[
  D_G(a) = \Big(\,h(a_{h^{-1}g})\,\Big)_{g,h \in G}\,. 
  \]

The following statement is useful in the sequel.
\begin{proposition}\label{prop:adjunction_right_composition}
  For any $a, b, c \in \LL [G]$, we have
  \[
    \langle a \circ \tau (b), c \rangle_{\LL[G]} = \langle
    a, c \circ b \rangle_{\LL[G]}.
  \]
\end{proposition}

\begin{proof}
  According to the description of $G$--Dickson matrices in
  Section~\ref{subsec:multiplication_map}, the maps
  \[
    \left\{
    \begin{array}{ccc}
      \LL[G]& \longrightarrow& \LL[G] \\
      x & \longmapsto & x \circ b
    \end{array}
  \right.
  \qquad {\rm and} \qquad
  \left\{
    \begin{array}{ccc}
      \LL[G]& \longrightarrow& \LL[G] \\
      x & \longmapsto & x \circ \tau(b)
    \end{array}
    \right.
  \]
  are represented in the canonical basis of $\LL[G]$ by the $G$--Dickson
  matrices $D_G(b)$ and $D_G (\tau (b))$.  From
  Lemma~\ref{lem:transpose_Dickson}, these matrices are transpose to
  each other. Thus, since the elements of $G$ form an orthonormal basis
  with respect to $\langle \cdot , \cdot \rangle_{\LL [G]}$,
  the corresponding maps are adjoint to each other.
\end{proof}

  \subsection{Support} First recall a very classical fact in
  adjunction which is that, given $a \in \LL [G]$, then we have
  $\ker(a)^\perp = \textrm{Im}(\tau(a))$, where the dual is taken with respect to $\langle \cdot, \cdot \rangle_{\mathrm{tr}}$. Now, let us introduce the
  notion of support of an element of $\LL[G]$.

\begin{definition}
  The {\em support} of an element $a \in \LL [G]$ is defined as
  the orthogonal of $\ker(a)$ with respect to $\langle \cdot, \cdot \rangle_{\mathrm{tr}}$. Namely,
  \[
    \supp (a) \mydef \ker(a)^\perp = \textrm{Im}(\tau (a)) \subseteq \LL.
  \]
\end{definition}

This definition can appear to be slightly different from the usual one
as given for instance in \cite[\S~2]{gorla2019} for matrix codes,
where the support of a matrix is its column space. However, our
definition can be understood as a row space. Indeed, the support
$\textrm{Im}(\tau(a))$ of $a$ can be interpreted as the column space
of a matrix representing $\tau(a)$ and hence as the row space of a
matrix representing $a$.  In particular, we have that
$\dim_\K(\supp(a)) = \rk(a)$.

Finally, let us recall the notion of {\em shortening} which is for instance
introduced in \cite[Definition 3.2]{byrne2017covering} (see also \cite[Definition~14]{sheekey2019mrd}).

\begin{definition}
  Let $\C \subseteq \LL[G]$ be a code and $I$ be a $\K$--subspace of $\LL$.
  The {\em shortening} of $\C$ at $I$ is defined as
  \[
    \Short_I(\C) \mydef \left\{c \in \C ~|~ I \subseteq \ker(c) \right\}.
  \]
\end{definition}

\subsection{Error correcting pairs}
The product of two codes $\A, \B \subseteq \LL[G]$ is defined as:
\[
\B \circ \A \mydef \Span{\LL}{b \circ a \mid a \in \A, b \in \B}\,.
\]
Notice that, generally, $\A \circ \B \ne \B \circ \A$. 
Next, given two codes $\A, \B \subseteq \LL[G]$ and some
$e \in \LL[G]$, we define
\[
\calK(e) \mydef \{ a \in \A \mid \langle b \circ a , e \rangle_{\LL[G]} = 0, \forall b \in \B \} \subseteq \LL[G]\,.
\]
Then we have the following result.
\begin{proposition}
  \label{prop:before-decoding-ECP}
  Let $\A, \B, \C \subseteq \LL[G]$ be codes such that $\B \circ \A \subseteq \C^\perp$. Let $r = c + e \in \LL[G]$, where $c \in \C$ and $e \in \LL[G]$. Denote $I = \supp(e)$. Then,
  \begin{enumerate}
  \item $\calK(r) = \calK(e)$,
  \item $\Short_{I}(\A) \subseteq \calK(e)$,
  \item if $\rk(e) < {\rm d}(\B^\perp)$, then $\Short_{I}(\A) = \calK(e)$.
  \end{enumerate}
\end{proposition}

\begin{proof}~
  \begin{enumerate}
  \item This holds since
    $\langle b \circ a, r\rangle_{\LL[G]} = \langle b \circ a,
    c\rangle_{\LL[G]} + \langle b \circ a, e \rangle_{\LL[G]}$ and,
    from $b \circ a \in \B \circ \A \subseteq \C^\perp$ we have
    $\langle b \circ a, c \rangle_{\LL[G]} = 0$.
  \item Let $a \in \Short_{I}(\A)$. Then $I \subseteq \ker(a)$ and
    hence
    $\ker(a)^\perp = \textrm{Im}(\tau(a)) \subseteq I^\perp
    = \ker(e)$, and hence $e \circ \tau(a) = 0$. Thus, from
    Corollary~\ref{prop:adjunction_right_composition}, we get
    $\langle b \circ a, e \rangle_{\LL[G]} = \langle b, e \circ
    \tau(a) \rangle_{\LL[G]} = 0$ for any $b \in \B$.
  \item Assume $\rk(e) \le {\rm d}(\B^\perp)$, and let us prove that
    $\calK(e) \subseteq \Short_{I}(\A)$. If $a \in \calK(e)$, then we
    have $e \circ \tau(a) \in \B^\perp$ by definition of
    $\calK(e)$. Since $\rk(e) < {\rm d}(\B^\perp)$, necessarily
    $e \circ \tau(a) = 0$ which yields
    $\ker(a)^\perp = \textrm{Im}\tau(a) \subseteq \ker(e) = I^\perp$, or equivalently,
    $I \subseteq \ker(a)$.
  \end{enumerate}
\end{proof}

We are now able to introduce error-correcting pairs in the context of
codes in $\LL[G]$. The definition is identical to the one given by
Mart{\'{\i}}nez{-}Pe{\~{n}}as and Pellikaan~\cite{Martinez-PenasP17}
in the context of rank-metric codes over finite fields.

\begin{definition}
  Let $\A, \B, \C \subseteq \LL[G]$ be three codes. The pair $(\A, \B)$ is a \emph{$t$-error-correcting pair} for $\C$ if the following holds:
  \begin{enumerate}
  \item $\B \circ \A \subseteq \C^\perp$,
  \item $\dim_\LL(\A) > t$,
  \item ${\rm d}(\B^\perp) > t$,
  \item ${\rm d}(\A) + {\rm d}(\C) > |G|$.
  \end{enumerate}
\end{definition}

Before showing how a $t$-error-correcting pair for a code $\C \subseteq \LL[G]$ enables to decode errors of rank up to $t$, we need a couple of technical lemmas.

\begin{lemma}
  \label{lem:system-solving}
  Let $(a_i)_i, (b_i)_i \in \LL[G]^M$ and $(c_i)_i \in \LL^M$. The system of $\K$-linear equations
  \[
  \langle a_i  \circ x, b_i \rangle_{\LL[G]} = c_i, \quad i=1, \dots, M,
  \]
   with unknown $x \in \LL[G]$, can be solved in $O(\min(M,N)MN^4)$ operations over $\K$, where $N = [\LL:\K]$.
\end{lemma}

\begin{proof}
  Let us fix a basis $(\beta_1, \dots, \beta_N)$ of $\LL/\K$. One
  writes
  $x = \sum_{g \in G} \sum_{j=1}^N x^{(j)}_g \beta_j g \in \LL[G]$,
  where $x_g^{(j)} \in \K$. Then, we have
  \[
    \langle a_i \circ x, b_i \rangle_{\LL[G]} = \sum_{g,h \in G}
    a_{i,g} g(x_{h}) b_{i,gh} = \sum_{g,h \in G} \sum_{j=1}^N a_{i,g}
    b_{i,gh} \, g(\beta_j) \, x_{h}^{(j)}\,.
  \]
  If we set
  $u^{(j)}_{i,h} = \sum_{g \in G} a_{i,g} b_{i,gh} g(\beta_j)$, then
  we end up with the system of $\K$-linear equations
  \[
    \sum_{j = 1}^N \sum_{h \in G} u^{(j)}_{i,h} x_h^{(j)} = c_i ,\quad
    i = 1, \dots, M,
  \]
  where $u^{(j)}_{i,h} \in \LL$, $c_i \in \LL$ and $x_h^{(j)} \in
  \K$. Using any basis of $\LL/\K$, these $M$ equations can be written
  as $MN$ equations over $\K$, with $N^2$ unknowns $\{ x_h^{(j)}
  \}$. Classical linear algebra algorithms solve this problem in
  $O(\min (MN, N^2) MN^3) = O(\min (M, N) MN^4)$ operations over $\K$.
\end{proof}

\begin{lemma}
  \label{lem:unique-c}
  Let $c \in \C$ and $r = c+e \in \LL[G]$, where
  $\supp(e) \subseteq J$ for some $\K$-vector space $J \subseteq \LL$
  such that $\dim_\K(J) < {\rm d}(\C)$. Then, $c$ is the unique
  element in $\C$ such that $\supp(r-c) \subseteq J$.
  Moreover, the codeword $c$ can be found by solving a system of
  linear equations over $\K$, with $O(N^3)$ equations and $O(N^2)$
  unknowns in $\K$, where $N = [\LL:\K]$.
\end{lemma}

\begin{proof}
  Assume $c, c' \in \C$ satisfy $\supp(r-c) \subseteq J$ and
  $\supp(r-c') \subseteq J$. Then,
  \[
    \supp(c-c') = \textrm{Im}(\tau(c-c')) \subseteq \textrm{Im}(\tau(r-c'))+
    \textrm{Im}(\tau(r-c)) \subseteq J.
  \]
  If $c \ne c'$, then
  $\dim_\K \supp(c-c') = \rk(c-c') \ge {\rm d}(\C)$ and we obtain a
  contradiction. Thus, $c=c'$.

  In order to compute $c$, it suffices to solve the system of $\K$-linear equations
  \[
    \left\{
    \begin{array}{l}
      \langle c, u_i \rangle_{\LL[G]} = 0, \\
       (r-c)(w_k)  = 0,
      \end{array}\right.
  \]
  where $\{ u_i \}_i$ is an $\LL$-basis of $\C^\perp$ and
  $\{ w_k \}_k$ is a $\K$-basis of $J^\perp$. One gets a system of
  $O(N^2)$ equations of the form given
  in~Lemma~\ref{lem:system-solving} which yields the result.
\end{proof}

\begin{theorem}
  Assume that $(\A, \B)$ is a $t$-error-correcting pair for
  $\C \subseteq \LL[G]$, where $2t+1 \le {\rm d}(\C)$. Then, there
  exists a deterministic algorithm $\mathsf{Dec}$ which runs in
  $O(N^7)$ operations over $\K$ given as input $r = c + e$ where
  $c \in \C$ and $e \in \LL[G]$ satisfies $\rk_\LL(e) \le t$, outputs
  the codeword $c$.
\end{theorem}

\begin{proof}
  Given $r = c + e$, the algorithm first computes $\calK(r)$; it
  consists of solving the system of equations
  \[
  \langle b_i \circ  x, r \rangle_{\LL[G]} = \langle x, a_j \rangle_{\LL[G]} = 0,
  \]
  where the unknown is $x \in \LL[G]$ and where $\{ b_i \}$ is an  $\LL$--basis of $\B$ and $\{ a_j \}$ is an $\LL$--basis of $\A^\perp$. 
  This can be done in $O(N^6)$ operations over $\K$ by Lemma~\ref{lem:system-solving}.

  Denote $I = \supp(e)$. Since $\calK(r) = \calK(e) = \Short_{I}(\A)$ by
  Proposition~\ref{prop:before-decoding-ECP}, one can now take an
  arbitrary nonzero element $a \in \mathcal K(r)$. Define
  $J = \ker(a) = \supp(a)^\perp$ and notice that $J$ contains
  $I$. Using the last condition in the definition of error-correcting
  pairs, we get
  \[
  \dim_\K(J) = |G| - \rk(a) \le |G| - {\rm d}(\A) < {\rm d}(\C)\,.
  \]
  Thus, from Lemma~\ref{lem:unique-c} one can find $c$ by solving another system of linear equations, requiring $O(N^7)$ operations over $\K$.
\end{proof}

\section{The abelian case: \texorpdfstring{$\theta$}{θ}-polynomials}
\label{sec:thetapoly}

In this section, we assume that
\[
  G = \Gal(\LL/\K)=\langle \theta_1, \ldots, \theta_m\rangle \cong
  \Z/n_1\Z \times \Z/n_2\Z \times \cdots \times \Z/n_m\Z\,.
  \]
  From now, we will also write elements of $\LL[G]$ with uppercase characters, \emph{e.g.} $P \in \LL[G]$, since they will be viewed as polynomials.

\subsection{Definition}\label{ss:Abelian_definitions}
  
Multivariate linearized polynomials can be defined as follows.  Let
$\bftheta=(\theta_1,\ldots, \theta_m)$ be a vector of generators of $G$. For a given
$\bfi = (i_1,\ldots, i_m) \in \N^m$, we denote by $\bftheta^\bfi$ the
element $\theta_1^{i_1} \circ \cdots \circ \theta_m^{i_m} \in G$ and
we write $|\bfi| \mydef i_1+\cdots+i_m$. Since
$\theta_i^{n_i} = \theta_i^0 = \mathrm{Id}$, we can actually consider
only tuples $\bfi$ belonging to
$\Delta(\bfn) \mydef \Delta(n_1) \times \cdots \times \Delta(n_m)$,
where $\Delta(t) \mydef \{0,1, \ldots, t-1\}$ and
$\bfn \mydef (n_1,\ldots, n_m)$. In this way, we have that
$G = \{\bftheta^\bfi \mid \bfi \in \Delta(\bfn)\}$ and hence, every
$P\in \LL[G]$ has a unique representation as
$$P = \sum_{\bfi \in \Delta(\bfn)} \bb_\bfi \bftheta^\bfi.$$
We also define $\bfone \mydef (1, \dots, 1) \in \N^m$. This will be used in Sections~\ref{sec:thetapoly} and~\ref{sec:reedmuller}.

\begin{definition}\label{def:thetapoly}
  A $\bftheta$-polynomial is an element
  $P = \sum_{\bfi \in \Delta(\bfn)} \bb_\bfi \bftheta^\bfi$
  belonging to the skew group algebra
  $\LL[G]=\LL[\theta_1,\ldots,\theta_m]$.  If $P$ is non-zero, then the $\bftheta$-degree of $P$ is the quantity
  $$\deg_{\bftheta}(P)  \mydef
  \max \{ |\bfi| \mid \bfi \in \Delta(\bfn), \bb_\bfi \neq 0 \}.$$
\end{definition}

Observe that $\bftheta$-polynomials are just elements of $\LL[G]$,
endowed with a notion of degree. This notion will be useful for
defining $\bftheta$-Reed--Muller codes and bounding their minimum
distance.

\subsection{Alon--F{\"u}redi Theorem and Schwartz--Zippel Lemma for \texorpdfstring{$\theta$}{θ}-polynomials}

In this section we show that we have an analogue of the celebrated
Alon--F{\"u}redi Theorem \cite[Theorem 5]{alon1993covering} and
Schwartz--Zippel Lemma \cite[Corollary 1]{sc80}.

Let $m$ and $N$ be positive integers with $m\leq N$ and let
$\bfa = (a_1,\ldots,a_m) \in \N^m$ be a vector of positive integers.
Define the integer $f(\bfa,N)$ as
 $$f(\bfa,N) \mydef\min\left\{ \prod_{i=1}^m b_i \mid
   \bfb-\bfone\in \Delta(\bfa) \mbox{ and } |\bfb|= N  \right\}.$$

\begin{lemma}\label{lem:minimumproduct}\cite[Lemma 2.2]{clark2017warning}
  Suppose $a_1\geq a_2 \geq \ldots \geq a_m$. Let $N\in \mathbb N$ be
  such that $N-m=\sum_{i=1}^s(a_i-1)+ \ell$ for some
  $s \in \{0,\ldots, m\}$ and $\ell$ such that $0\leq \ell <
  a_{s+1}$. Then \begin{equation}
    \label{eq:min-dist}f(\bfa,N)=(\ell+1)\prod_{i=1}^sa_s.
    \end{equation}
\end{lemma}
We recall now the classical versions of Alon-F{\"u}redi Theorem and
Schwartz-Zippel Lemma. For this purpose, we introduce the following
notation. Let $\F$ be a field and let $S\subseteq \F^m $ be a fixed
set. Moreover, let $p\in\F[x_1,\ldots,x_m]$ be a multivariate
polynomial. We denote by $U_S(p)$ and $V_S(p)$ the set of non-zeros
and of zeros, respectively, of $p$ in $S$, that is
$$U_S(p):=\left\{\bfu \in S \mid p(\bfu)\neq 0\right\}, \quad V_S(p):=\left\{ \bfv\in S \mid p(\bfv)=0\right\}.$$

\begin{theorem}[Alon--F\"uredi Theorem]\cite[Theorem 5]{alon1993covering}
  Let $S=S_1\times \cdots \times S_m\subseteq \F^m$ be a finite grid
  with $S_i\subseteq \F$ and $|S_i|=n_i$, where
  $n_1 \geq n_2 \geq \cdots \geq n_m \geq 1$. Let
  $p \in \F[x_1,\ldots,x_m]$ be a polynomial that is not identically
  $0$ on $S$, and let $\bar{p}$ be the polynomial $p$ modulo the ideal
  $(p_1(x_1),\ldots, p_m(x_m))$, where
  $p_i(x_i)=\prod_{s\in S_i}(x_i-s)$. Then
$$|U_S(p)| \geq (n_s-\ell)\prod_{i=1}^{s-1} n_i.$$ 
  where $\ell$ and $s$ are the unique integers satisfying
  $\deg \bar{p}=\sum_{i=s+1}^k (n_i-1) + \ell$, with $1\leq s \leq k$ and $1 \leq \ell <n_s$. 
\end{theorem}
 
\begin{lemma}[Schwartz--Zippel Lemma]\cite[Corollary 1]{sc80}.
  Let $S=S_1\times \cdots \times S_m\subseteq \F^m$ be a finite grid
  with $S_i\subseteq \F$ and $|S_i|\geq 1$ for each
  $i\in\{1,\ldots,m\}$.  Let $p \in \F[x_1,\ldots,x_k]$ be a nonzero
  polynomial. Then,
  $$|V_S(p)| \leq \frac{\deg (p)}{\min\{|S_1|,\ldots,|S_m|\}}|S|.$$ 
\end{lemma}

At this point, we are ready to state the Alon--F\"uredi Theorem for
$\bftheta$-polynomials, which is the central result of this section.

\begin{theorem}[Alon--F{\"u}redi  Theorem for $\bftheta$-polynomials]\label{thm:AF}
  Let $\bfn = (n_1, \dots, n_m)$ be an $m$-tuple of non-negative
  integers such that $n_1 \geq n_2 \geq \cdots \geq n_m \geq 2$ and
  let
  $G = \langle \theta_1, \dots, \theta_m \rangle \simeq \Z/n_1\Z
  \times \dots \times \Z/n_m\Z$ be the Galois group of a field
  extension $\LL/\K$. Moreover, let $P \in \LL[G]$ be nonzero. Then
$$\rk(P) \geq (n_s-\ell)\prod_{i=1}^{s-1} n_i.$$ 
  where $\ell$ and $s$ are the unique integers satisfying
  $\deg_\bftheta(P)=\sum_{i=s+1}^m (n_i-1) + \ell$, with $0 \leq \ell <n_s$.
\end{theorem}

\begin{proof}
  Let
  $P = \sum_{\bfi \in \Delta(\bfn)} \bb_\bfi \bftheta^\bfi$ be a $\bftheta$-polynomial with $\deg_\bftheta(P)=\sum_{i=s+1}^m (n_i-1) + \ell$.
  Our goal is to find an ordering of $G$ for which $\rk(P) =
  \rk_\LL(D_G(P))$ can be easily bounded.  Let us fix a monomial order
  $\prec$ on $\mathbb{N}^m$, which is a refinement of the total degree, that is, for each finite set $\mathcal S\subset \mathbb{N}^m$, the maximal element in $\mathcal S$ with respect to $\prec$ has also maximal total degree among the elements of $\mathcal S$. 
    We write the group
  \[
    G=\left\{\bftheta^{\bfi^{(1)}},\ldots, \bftheta^{\bfi^{(N)}}\right\} 
  \]
  according to the order $\prec$ restricted to
  $\Delta(\bfn)=\{\bfi^{(1)},\ldots,
  \bfi^{(N)}\}$. We also denote
  $\mathrm{lt}_{\prec}(P)=\bftheta^{\bfi^{(s)}}$ the leading term of
  $P$, and
  $\mathrm{lc}_{\prec}(P)=\bb_{\bfi^{(s)}}$ its leading coefficient,
  for some $\bfi^{(s)} = (u_1,\ldots,u_m) \in \Delta(\bfn)$.

  Consider the $G$--Dickson matrix $D_G(P)$ with respect to this
  order on $G$.  In the first column of $D_G(P)$, the $(s,1)$-entry is
  $\bb_{\bfi^{(s)}} \neq 0$, and the $(j,1)$-entry is $0$ for every
  $s<j \leq N$. Let us define
  \[
    \mathcal{T} \mydef \{\bftheta^\bfv \mid v_i < n_i-u_i\} \subseteq
    G\quad {\rm and}\quad t \mydef |\mathcal{T}| =
    \prod_{i=1}^m(n_i-u_i).
  \]
  We also order
  $\mathcal{T} = \{\bftheta^{\bfj^{(1)}}, \ldots,
  \bftheta^{\bfj^{(t)}}\}$ according to $\prec$, \emph{i.e.}
  $\bftheta^{\bfj^{(1)}}\prec \cdots \prec \bftheta^{\bfj^{(t)}}$.

  Let us now fix $i \in \{1, \dots, t\}$. By definition, the column of
  $D_G(P)$ corresponding to $\bftheta^{\bfj^{(i)}} \in \mathcal{T}$ is
  given by the coordinates of the $\bftheta$-polynomial
  $\bftheta^{\bfj^{(i)}} \circ P$ in the basis
  $\{\bftheta^{\bfi^{(1)}},\ldots, \bftheta^{\bfi^{(N)}}\}$.  We have
  that
  $\mathrm{lt}_\prec(\bftheta^{\bfj^{(i)}} \circ P) =
  \bftheta^{\bfj^{(i)}+\bfi^{(s)}} = \bftheta^{\bfi^{(s_i)}}$, for a
  suitable positive integer $s_i\leq N$. Moreover, by definition of a
  monomial order, we have $s=s_1<s_2<\cdots <s_t\leq N$, and in the
  column corresponding to $\bftheta^{\bfj^{(i)}}$, all the elements
  with row index $j$ for $s_i<j\leq N$ are equal to $0$. Furthermore,
  the element with row index $s_i$ equals
  $\bftheta^{\bfj^{(i)}}(\bb_{\bfi^{(s)}}) \neq 0$. Therefore, the
  submatrix $D_{\mathcal{T}}$ of $D_G(P)$ obtained by taking the columns
  corresponding to $\mathcal{T}$ and the rows $s_1,\ldots,s_t$, is
  an upper triangular $t\times t$ matrix of the form
  \[
    D_{\mathcal{T}} = \begin{pmatrix} \bftheta^{\bfj^{(1)}}(\bb_{\bfi^{(s)}}) &  &  &
       &  \\
      & \bftheta^{\bfj^{(2)}}(\bb_{\bfi^{(s)}}) &  & (*) & \\
      & & \ddots &  & \\
      &(0) & &  \bftheta^{\bfj^{(t-1)}}(\bb_{\bfi^{(s)}}) &  \\
      & & & & \bftheta^{\bfj^{(t)}}(\bb_{\bfi^{(s)}}) \end{pmatrix},
  \]
  with nonzero elements on the diagonal.
  Hence, by Theorem \ref{thm:rankcharacterization}, we have 
  \[
  \rk(P)=\rk_{\LL}(D_G(P)) \geq \rk_{\LL}(D_{\mathcal{T}})=|\mathcal{T}| =
  \prod_{i=1}^m(n_i-u_i).
  \]
  We conclude the proof by observing that
  \begin{align*}
 f\Big(\bfn,\Big(\sum_{i=1}^m n_i\Big) - \deg_\bftheta(P)\Big) &= \min \left\{ \prod_{i=1}^m v_i ~\bigg|~ \bfv-\bfone 
       \in \Delta(\bfn), |\bfv|= \Big(\sum_in_i\Big) - \deg_\bftheta(P) \right\} \\
 &=\min\left\{ \prod_{i=1}^m (n_i-u_i) ~\bigg|~ \bfu  \in \Delta(\bfn), |\bfu|= \deg_\bftheta(P)\right\}.
        \end{align*}
        and from Lemma \ref{lem:minimumproduct} we get the desired result.\\
\end{proof}

\begin{remark}
  From the proof of Theorem \ref{thm:AF} one can easily see that the result can be refined if we make further assumptions on the element $P\in \LL[G]$. Indeed, if there exists a monomial order $\prec'$ on $\mathbb{N}^m$ such that $\mathrm{lc}_{\prec'}(P)=\bftheta^\bfu$ with $|\bfu|<\deg_\bftheta(P)$, using the same proof with the monomial order $\prec'$, one gets that
  $$ \rk(P)\geq f\Big(\bfn,\Big(\sum_{i=1}^m n_i\Big) - |\bfu|\Big). $$
  In the Hamming metric, the effects of the choice of monomial orders for designing codes with better minimum distance have been intensively studied by Geil and Thomsen in \cite{geil2013weighted}.
\end{remark}




Actually, the  Alon--F{\" u}redi Theorem for $\bftheta$-polynomials allows to prove  an analogue in the rank metric of the well-known Schwartz--Zippel lemma. This can be stated as follows.
\begin{corollary}[Schwartz--Zippel Lemma for $\bftheta$-polynomials]\label{cor:SZ}
  Let  $\bfn = (n_1, \dots, n_m)$ be an $m$-tuple of non-negative integers,
  let $G = \langle \theta_1, \dots, \theta_m \rangle \simeq \Z/n_1\Z \times \dots \times \Z/n_m\Z$ be the Galois group of a field extension $\LL/\K$, and let $P \in \LL[G]$. Then, we have:
  \[
  \dim_{\K} \ker(P) \; \le \; \frac{\deg_\bftheta(P)}{\min \{ n_1, \dots, n_m \}} \cdot \prod_{i=1}^m n_i\,. 
  \]
\end{corollary}

\begin{proof}
  Without loss of generality, we can assume
  $m_1 \geq n_2 \geq \cdots \geq n_m \geq 2$, so that
  $\min\{n_1,\ldots,n_m\}=n_m$. If $\deg_\bftheta(P)\geq n_m$ there is
  nothing to prove. Hence, suppose $\deg_{\theta}<n_m$. Using Theorem
  \ref{thm:AF} we obtain
  $$ \dim_{\K} \ker(P)=\prod_{i=1}^mn_i-\rk(P)\geq \prod_{i=1}^mn_i-(n_m-\deg_\bftheta(P))\prod_{i=1}^{m-1}n_i=\deg_\bftheta(P)\prod_{i=1}^{m-1}n_i.$$
\end{proof}

\section{\texorpdfstring{$\theta$}{θ}-Reed--Muller codes}
\label{sec:reedmuller}

In this section, we introduce and develop the theory of
$\bftheta$-Reed--Muller codes. They can be seen either as the
counterparts of Reed--Muller codes in the rank metric, or as the
multivariate version of Gabidulin codes.

\subsection{Definition}

We assume to work in the setting described in Section \ref{sec:thetapoly}.

\begin{definition}
  Let $\LL/\K$ be a Galois extension such that
  $G \mydef \Gal(\LL/\K) = \langle \theta_1,\ldots,\theta_m\rangle
  \cong \Z/n_1\Z \times \cdots \times \Z/n_m\Z$ and let
  $r\in \mathbb{N}$ such that $r\leq \sum_i (n_i-1)$.  The
  \emph{$\bftheta$-Reed--Muller code of order $r$ and type $\bfn$} is
  \[
  \mathrm{RM}_\bftheta(r,\bfn) \mydef \left\{ P \in \LL[G] \mid
  \deg_\bftheta(P) \leq r \right\} \subseteq \LL[G].
  \]
\end{definition}

\begin{remark}
  The definition of $\bftheta$-Reed--Muller codes depends on the
  choice of generators $\bftheta$ of the Galois group $G$.  This is
  somehow similar to the case of (generalized) Gabidulin codes.
\end{remark}

\begin{remark}
  Given a basis $\B$ of $\LL/\K$, the vectorial version of
  $\mathrm{RM}_\bftheta(r,\bfn)$ is then
  \[
  \mathrm{RM}_{\bftheta, \B}(r,\bfn)\mydef \RM_\bftheta(r, \bfn)(\B) =
  \left\{ \mathrm{ev}_{\B}(P) \mid P \in \LL[G], \deg_{\bftheta}(P)
  \leq r \right\}\subseteq \LL^N,
\]
where $\mathrm{ev}_{\B}(P)$ is the evaluation vector as defined
in~(\ref{eq:ev_map}).
\end{remark}

\begin{example}
  Let $\K = \mathbb{Q}(\zeta)$ where $\zeta^2 + \zeta + 1 =
  0$. Consider $\LL/\K$ a Galois extension of degree $6$ given by
  $\LL = \K(\sqrt{p}, \sqrt[3]{q})$, where $p$ and $q$ are two
  distinct primes. Then
\[
  \B = \Big( 1,\; \sqrt{p},\; \sqrt[3]{q},\; \sqrt{p}\sqrt[3]{q},\;
  \sqrt[3]{q^2},\; \sqrt{p}\sqrt[3]{q^2} \Big) \in \LL^6
\]
is an ordered $\K$-basis of $\LL$. Moreover we have
$G = \Gal(\LL/\K) = \langle \theta_1, \theta_2 \rangle$ where
\[
\theta_1 : \left\{ \begin{array}{ll} \sqrt{p} &\mapsto (-1)
  \cdot \sqrt{p}\\
  \sqrt[3]{q} &\mapsto 1 \cdot
  \sqrt[3]{q}
\end{array}\right.
\quad \text{ and } \quad
\theta_2 : \left\{  \begin{array}{ll} \sqrt{p} &\mapsto 1 \cdot \sqrt{p}\\
                      \sqrt[3]{q} &\mapsto \zeta \cdot \sqrt[3]{q}
\end{array}\right. .
\]
We observe that $\theta_1^2 = \theta_2^3 = {\Id}$, hence
$\bfn = (n_1, n_2) = (2, 3)$ and $N = |\Gal(\LL/\K)| = n_1 n_2 = 6$.

Let now $r = 1$. The $(\theta_1, \theta_2)$-Reed--Muller code of
order $r$ is
\[
\RM_\bftheta(r, \bfn) = \{ a \cdot {\Id} + b \cdot \theta_1 + c
\cdot \theta_2 \mid a,b,c \in \LL \} \subseteq \LL[G].
\]
Its vectorial version with respect to the basis
  $\B = (\bb_1, \dots, \bb_6)$ defined above, has the following generator
  matrix:
  \[
  \begin{pmatrix}
    \bb_1 & \bb_2 & \bb_3 & \bb_4 & \bb_5 & \bb_6\\
    \bb_1 & -\bb_2 & \bb_3 & -\bb_4 & \bb_5 & -\bb_6\\
    \bb_1 & \bb_2 & \zeta \bb_3 & \zeta \bb_4 & \zeta^2 \bb_5 & \zeta^2 \bb_6
  \end{pmatrix}\,.
  \]
\end{example}

\subsection{Parameters of \texorpdfstring{$\theta$}{θ}-Reed--Muller codes}

We now compute the dimension and the minimum rank distance of
$\bftheta$-Reed--Muller codes.

\begin{proposition}
  \label{prop:dimension-thetaRM}
  The dimension of $\mathrm{RM}_\bftheta(r,\bfn)$ is equal to the
  cardinality of the set
  $\{ \bfi \in \Delta(\bfn) \mid |\bfi|\leq r \},$ that in turn is
  equal to
  $$k(r,\bfn)=\sum_{\ell=0}^r c(\ell,\bfn)=\sum_{\ell=0}^r
  [z^\ell]\prod_{j=1}^m\left(\frac{1-z^{n_j}}{1-z} \right),$$ where
  $c(\ell,\bfn)$ of the integer $\ell$ in at most $m$ parts in which
  the $j$-th part is at most $n_j-1$ and $[z^\ell]p(z)$ denotes the
  coefficient of $z^\ell$ in the polynomial $p(z)$.
\end{proposition}

\begin{proof}
  By definition a set of generators for the $\bftheta$-Reed--Muller code
  is given by the set
  $\{\bftheta^\bfi \mid \bfi \in \Delta(\bfn), |\bfi|\leq r\}$. Moreover
  these $\bftheta$-monomials are linearly independent over $\LL$, by
  Artin's theorem. Therefore the dimension of the code is equal to the
  cardinality $k(r,\bfn)$ of the set
  $\{ \bfi \in \Delta(\bfn) \mid |\bfi|\leq r \}$.  Let $c(\ell,\bfn)$
  denote the number of weak compositions of the integer $\ell$ in at most
  $m$ parts in which the $j$-th part is at most $n_j-1$. Then,
  $$k(r,\bfn)=\sum_{\ell=0}^r c(\ell,\bfn).$$
  Since it is well-known that
  $c(\ell,\bfn)=[z^\ell]\prod_{j=1}^m\left(\frac{1-z^{n_j}}{1-z} \right)$,
  we can conclude.
\end{proof}

For every $i \in \{1,\ldots, m\}$ we also consider the subgroup
$G_i=\langle \theta_j \mid j \in \{1,\ldots, m\} \setminus \{i\}
\rangle$, and the corresponding fixed field
\[
\LL_i \mydef \LL^{G_i}= \left\{ a \in \LL \mid \sigma(a)=a,
\mbox{ for every } \sigma \in G_i \right\}.
\]

Before determining the minimum distance of $\bftheta$-Reed--Muller
codes, we define an object of particular interest in the case of
cyclic extensions.

\begin{proposition}\label{prop:annihiliatorpoly}\cite[Theorem 2]{augot2013rank}
  Let $\LL/\K$ be a cyclic Galois extension of degree $n$, with Galois
  group $G = \langle \theta\rangle$. Let
  $V \mydef \Span{\K}{ v_1,\ldots, v_r} \subseteq \LL$  be a
  $\K$-subspace of dimension $r\geq 0$. Then, there exists a unique
  monic $\theta$-polynomial $P_V \in \LL[\theta]$ of $\theta$-degree
  $r$ such that $P_V(V) = \{ 0 \}$. Moreover, the polynomial $P_V$
  is defined by induction as:
  \[
    P_V= \begin{cases} \Id & \mbox{ if } r=0\\
      \left(\theta- \frac{\theta(P_{V_1}(v_r))}{P_{V_1}(v_r)}\right)
      \circ P_{V_1} & \mbox{ if } r\geq1,
  \end{cases}
  \]
  where $V_1 \mydef \Span{\K}{ v_1,\ldots, v_{r-1}}$.
\end{proposition}

\begin{proof}
  The existence and uniqueness follows from the fact that
  $\LL[\theta]$ is a left Euclidean domain. In particular, the left
  ideal
  $I \mydef \{P \in \LL[\theta] \mid P(\vvv)=0 \mbox{ for every } \vvv \in V
  \}$ is principal. In addition, $I$ contains $P_V$.
  Moreover, it is well-known that the dimension of the kernel of a $\theta$--polynomial
  is bounded from above by its $\theta$--degree. This can be deduced, for instance, from Corollary \ref{cor:SZ}. Therefore, $P_V$
  is a monic element of $I$ of the least possible degree. Hence it is a generator
  of $I$.
  Moreover, the polynomial defined by the recursive formula has
  $\theta$-degree $r$, is monic and it annihilates the subspace $V$.
\end{proof}

The polynomial $P_V$ defined by Proposition
\ref{prop:annihiliatorpoly} is called the 
\emph{annihilator polynomial of the subspace $V$}. 
In the finite field case, this coincides with the 
notion of annihilator or subspace polynomial, which is
a linearized polynomial of degree $q^r$ whose roots are exactly the
elements of an $r$-dimensional $\Fq$-subspace of $\F_{q^n}$.

\begin{theorem}
  \label{thm:minimum-distance}
  Let $r$ be a positive integer and $\bfn=(n_1,\ldots, n_m) \in \N^m$
  be a vector such that $n_1 \geq n_2 \geq \cdots \geq n_m \geq
  2$. Then the minimum rank distance of the code
  $\mathrm{RM}_\bftheta(r,\bfn)$ is equal to
  \[
    d(r,\bfn)= \min\left\{ \prod_{i=1}^m (n_i-u_i) \mid \bfu =
      (u_1,\ldots, u_m) \in \Delta(\bfn), |\bfu|\leq r\right\}\,.
  \]
  In particular, $d(r,\bfn) = 1$ if $r \ge \sum_{i=1}^m (n_i -1)$, and
  otherwise
  \[
  d(r,\bfn) = (n_s-\ell) \prod\limits_{i=1}^{s-1} n_i\,
  \]
  where $\ell$ and $s$ are the unique integers satisfying
  $r=\sum_{i=s+1}^m (n_i-1) + \ell$, with $0 \leq \ell <n_s$.
\end{theorem}
\begin{proof}
  \emph{Lower bound.\;} First, it is easy to observe that the minimum
  is met for an element $\bfu$ such that $|\bfu|=r$. At this point,
  the lower bound directly follows from Theorem \ref{thm:AF}, since
  the minimum distance is the minimum rank of $P$ among all the
  nonzero $P\in \mathrm{RM}_\bftheta(r,\bfn)$ of $\bftheta$-degree
  equal to $r$.

  \noindent \emph{Upper bound.\;} 
Let now $r \ge 1$, and $\ell$, $s$
  the unique integers satisfying $r=\sum_{i=s+1}^m (n_i-1) + \ell$,
  with $1 \leq \ell <n_s$. For every $i \in \{s+1,\ldots,m\}$, choose a 
  $\K$-subspace $V_i$ of $\LL_i$ with dimension $n_i-1$ that does not contain $\K$, that is, $V_i\cap \K=\{0\}$. Moreover, choose $V_s$ to be any $\K$-subspace of $\LL_s$ of dimension $\ell$ that does not contain $\K$.
  For each $i\in \{s,\ldots, m\}$, let $P_i \in \LL_i[\theta_i]$ to be the annihilator $\theta_i$-polynomial of $V_i$.
  Observe that if $j \ne i$, then for every $x \in \LL_j$ we have $P_i(x) = P_i(1)x$. Thus, define $\tilde{P}_i:=P_i(1)^{-1}P_i$, and consider the $\bftheta$-polynomial
  $\tilde{P} \mydef \tilde{P}_s \circ \tilde{P}_{s+1}\circ \cdots \circ \tilde{P}_m$. 
  We then have $\tilde{P}(V_i)=0$ for every $i \ge s$.

  Given $j \in \{ s, \dots, m\}$ let us define
  \[
  \mathcal{U}_j \mydef \LL_1  \cdots  \LL_{j-1} 
  V_j =\LL_{(j-1)}V_j  \subseteq \LL,
  \]
  where for two $\K$-subspaces $W, W'$ of $\LL$, we define $W W' \mydef
  \Span{\K}{ w w' \mid w \in W, w' \in W'}$, and $\LL_{(j-1)}$ denotes the compositum of $\LL_1,\ldots, \LL_{j-1}$. Then, we see that for every $j \ge s$ we have  $\ker(\tilde{P}) \supseteq \mathcal{U}_j$ and
  $\mathcal U_j \cap \left(\mathcal U_{j+1}+\cdots + \mathcal
    U_m\right)=\{0\}$. Therefore,
  \[
    \rk(\tilde{P}) =\prod_{i=1}^m n_i- \dim(\ker(\tilde{P})) \leq
    \prod_{i=1}^m n_i- \sum_{j=s}^m \dim(\mathcal U_j).
  \]
  Since $\dim(\mathcal U_s) = \ell \prod_{i=1}^{s-1} n_i$ and
  $\dim(\mathcal U_j) = (n_j-1) \prod_{i=1}^{j-1}n_i$ for $j \ge s+1$,
  this yields
  \[
    \rk(\tilde{P}) \le \prod_{i=1}^m n_i- \sum_{j=s+1}^m
    (n_j-1)\prod_{i=1}^{j-1}n_i - \ell \prod_{i=1}^{s-1} n_i
  = (n_s-\ell)\prod_{j=1}^{s-1} n_j
  \]
  from which we get the desired upper bound.  
\end{proof}

\subsection{Duality}

In this section we study the duality properties of
$\bftheta$-Reed--Muller codes, showing that such a family is
essentially closed under duality (see Proposition
\ref{prop:dual1}). For this purpose, let us denote
$\bftheta_{\rm inv} = (\theta_1^{-1}, \dots, \theta_m^{-1})$. It is
clear that $\bftheta_{\rm inv}$ is also a system of generators for the
Galois group $G$.  Let us also denote
$\bftheta^{-\bfone} = \theta_1^{-1} \circ \dots \circ \theta_m^{-1}$.

\begin{proposition}
  \label{prop:dual1}
Let  $p = \sum_{j=1}^m (n_j-1)$. Then we have:
  \[
    \RM_{\bftheta}(r, \bfn)^\perp = \RM_{\bftheta_{\rm inv}}(p-r-1,
    \bfn)
    \circ\bftheta^{-\bfone}=\bftheta^{-\bfone}\circ\RM_{\bftheta_{\rm
        inv}}(p-r-1, \bfn)\,.
  \]
\end{proposition}

\begin{proof}
  It is clear that the dual of $\RM_{\bftheta}(r, \bfn)$ is the
  $\LL$-span of the set
  $$\left\{\bftheta^\bfi \mid \bfi \in \Delta(\bfn), i_1+\cdots+i_m>r \right\}.$$
  Observe that we can write
  $\bftheta^\bfi=(\bftheta^{-\bfone})^{\bfn-\bfone-\bfi}\circ\bftheta^{-\bfone}=\bftheta^{-\bfone}\circ
  (\bftheta^{-\bfone})^{\bfn-\bfone-\bfi}$. Moreover,
  $\bfi \in \Delta(\bfn)$ with $\sum_{j=1}^m i_j>r$ if and only if
  $\bfn-\bfone-\bfi \in \Delta(\bfn)$ with
  $\sum_{j=1}^m n_j-1-i_j \leq p-r-1$. This concludes the proof.
\end{proof}

Proposition~\ref{prop:dual1} can be translated in the vector setting
as follows.

\begin{corollary}
   \label{prop:dual}
   Let $\B$ be a basis of $\LL/\K$ and $p = \sum_{j=1}^m
   (n_j-1)$. Then we have:
   \[
     (\RM_\bftheta(r, \bfn)(\B))^\perp = \RM_{\bftheta_{\rm
         inv}}(p-r-1, \bfn)(\bftheta^{-\bfone}(\B^*))\,.
   \]
 \end{corollary}

\begin{proof}
  Combining Theorem \ref{thm:dualLG} and Proposition \ref{prop:dual1},
  we get that
  \begin{align*}
    \RM_{\bftheta}(r, \bfn)(\B)^\perp =
    \left(\RM_{\bftheta}(r, \bfn)^\perp\right)(\B^*) =
    \left(\RM_{\bftheta_{\rm inv}}(p-r-1, \bfn)
    \circ\bftheta^{-\bfone}\right)(\B^*).
  \end{align*}
  At this point one can observe that for every $P\in\LL[G]$, it holds
  $\ev_{\B^*}(P\circ
  \bftheta^{-\bfone})=\ev_{\bftheta^{-\bfone}(\B^*)}(P)$, giving
  \[
    \left(\RM_{\bftheta_{\rm inv}}(p-r-1, \bfn) \circ \bftheta^{-\bfone}\right)(\B^*)=\RM_{\bftheta_{\rm inv}}(p-r-1, \bfn)(\bftheta^{-\bfone}(\B^*)).
    \]
\end{proof}

\subsection{Decoding \texorpdfstring{$\theta$}{θ}-Reed--Muller codes}

In this section, we shortly explain how error-correcting pairs allow
to decode $\bftheta$-Reed--Muller codes up to some error weight. The
decoding capability is however non-optimal, and we leave open the
question of the decoding $\bftheta$-Reed--Muller codes up to half
their minimum distance.

The key point is to notice the following.
\begin{lemma}
  Let $r, r' \ge 0$ such that
  $r+r' \le p \mydef \sum_{i=1}^m (n_i - 1)$. Then we have
  \[
  \RM_\bftheta(r, \bfn) \circ \RM_\bftheta(r', \bfn) = \RM_\bftheta(r + r', \bfn)\,.
  \]
\end{lemma}
\begin{proof}
  This is clear since
  $\deg_\bftheta(\bftheta^\bfi\bftheta^\bfj) \le r+r'$ whenever
  $\deg_\bftheta(\bftheta^\bfi) \le r$ and
  $\deg_\bftheta(\bftheta^\bfj) \le r'$.
\end{proof}

We recall that $d(r, \bfn)$ and $k(r, \bfn)$ respectively represent
the minimum distance and the dimension of
$\bftheta$-Reed--Muller. Their definition are given in
Theorem~\ref{thm:minimum-distance} and
Proposition~\ref{prop:dimension-thetaRM}.

\begin{proposition}
  Let $r, t \ge 0$ and assume that $2t+1 \le d(r,
  \bfn)$ 
  Set $N = \prod_{i=1}^m n_i$ and $p \mydef \sum_{i=1}^m (n_i-1)$. Let
  $\A = \RM_{\bftheta_{\rm inv}}(a, \bfn)$ and
  $\B = \RM_{\bftheta_{\rm inv}}(b, \bfn)$ be such that
  \begin{enumerate}
  \item $a+b \le p - r - 1$,
  \item $k(a, \bfn) > t$,
  \item $d(p-1-b, \bfn) > t$,
  \item $d(a, \bfn) + d(r, \bfn) > N$.
  \end{enumerate}
  Then, $(\A, \B)$ is an error-correcting pair for
  $\C = \RM_\bftheta(r, \bfn) \circ \bftheta^{-\bfone}$.
\end{proposition}

\begin{proof}
  It follows from the definition of error-correcting pairs and the
  duality results from Proposition~\ref{prop:dual}.
\end{proof}

A natural question is to compute the maximum decoding radius $t$ one can get with a $t$-error correcting pair for a given code $\mathcal{C}_r = \RM_\bftheta(r, \bfn)$. In the following example, we initiate this study by considering the simplest non-trivial case $\bfn = (n, n)$, $n \ge 2$.

\begin{example}
    Let us fix $\bfn = (n,n)$ and $r \le 2n-3$. 
  For clarity let us also use the simpler notation $d(x) \mydef d(x, \bfn)$ and $k(x) \mydef k(x, \bfn)$. The goal is to find the maximum $t$ for which there exists a pair $(a,b)$ such that $(\RM_{\bftheta_{\rm inv}}(a, \bfn), \RM_{\bftheta_{\rm inv}}(b, \bfn))$ is a $t$-error-correcting pair for $\RM_\bftheta(r, \bfn)$. In other words, we look for
  \[
  t_{\rm max} = \max \Big\{ \min \{ k(a), d(2n-3-b) \} - 1 \;\Big|\; d(a) + d(r) \ge n^2+1 \text{ and } a+b \le 2n-3-r \Big\}.
  \]  
  In this context, we have
  \[
  d(x) = \left\{\begin{array}{ll} n^2-nx & \text{if } 0 \le x \le n-1 \\ 2n-1-x & \text{if }n \le x \le 2n-2 \end{array}\right.
  \]
  and
  \[
  k(x) = \left\{\begin{array}{ll} \frac{(x+1)(x+2)}{2} & \text{if } 0 \le x \le n-1 \\ n^2 - \frac{(2n-1-x)(2n-2-x)}{2} & \text{if } n \le x \le 2n-2. \end{array}\right.
  \]
  Maps $d$ and $k$ are illustrated in Figure~\ref{fig:radii}.
  
  If $r \ge n-1$, then $d(r) = 2n-1-r$ and one needs to set $a = 0$ to
  fulfill the condition $d(a) + d(r) \ge n^2+1$. Thus
  $t_{\rm max} = 0$, which means that $\RM_\bftheta(r, \bfn)$ admits no
  non-trivial error-correcting pair of the desired form.

  Therefore, let us consider the more interesting case $r \le n-2$.
  Define $u = 2n-3-a-b$. Since $d(r) = n^2-nr$, we have
  \[
  t_{\rm max} = \max \Big\{ \min \{ k(a), d(a+u) \} - 1 \;\Big|\; d(a) \ge nr + 1 \text{ and } u \ge r \Big\}.
  \]

  For any fixed $a$, the map $u \mapsto \min \{ k(a), d(a+u) \}$ is decreasing, therefore $t_{\rm max}$ is reached for $u = r$. Moreover, the condition $d(a) \ge nr +1$ is equivalent to $a \le n - r -1$. We also see that $d(\cdot)$ is decreasing and $k(\cdot)$ is increasing, thus $t_{\rm max} = \min \{ k(\lfloor \alpha \rfloor), d(\lceil \alpha \rceil + r) \}-1$  where $\alpha \in [0,n-1-r]$ is the only real number satisfying $k(\alpha) = d(\alpha+r)$. A simple computation shows that $\alpha = -n - \frac{3}{2} + \sqrt{3n^2 +(3-2r)n + \frac{1}{4}}$.

  Asymptotically, let us set $\rho = \lim_{n \to \infty} \frac{r}{n}$. Then we see that $\alpha = (\sqrt{3-2\gamma}-1)n + O(\sqrt{n})$, hence $t_{\rm max} = (2 - \gamma - \sqrt{3 - 2\gamma}) n^2 + O(n^{3/2})$. It means that the corresponding error-correcting pair can correct approximately $(2 - \gamma - \sqrt{3 - 2\gamma})n^2$ errors, while the unique decoding radius of $\RM_\bftheta(r, \bfn)$ is $\lfloor \frac{d(r)-1}{2} \rfloor \simeq \frac{1-\gamma}{2}n^2$. See Figure~\ref{fig:radii} for a comparison.

  \begin{figure}[h!]
    \begin{tikzpicture}[scale=0.75]      

      \def\n{13}
      
        \pgfplotsset{every tick label/.append style={font=\normalsize}}
        
        \begin{axis}[
          xmin=0,
          xmax={2*\n-2},
          ymin=0,
          ymax={\n *\n},
          xlabel={$x$},
          ylabel={$d(r), k(r)$},
          ytick={0, 13, 169},
          yticklabels={$0$, $n$, $n^2$},
          xtick={0, 12, 24},
          xticklabels={$0$, $n-1$, $2n-2$},
          xlabel style={anchor=north, at={(0.5,-0.08), font=\normalsize}},
          ylabel style={anchor=south, at={(-0.02,0.5), font=\normalsize}},
          legend style= {anchor = east, at={(0.95,0.5), font=\small}},
          cycle list name=mark list*
          ]
          
          \tikzstyle{my_style}=[domain=0:{2*\n-2}, mark=none, samples=100]

          \addplot+[my_style, color=green!60!black]
          {(x <= \n-1) *(\n*\n-\n*x) + (x >= \n-1) * (2*\n-1-x)}; 
          \addlegendentry{$d(r)$}

          \addplot+[my_style, color=blue]
          {(x <= \n-1) * ((x+1)*(x+2)/2) + (x >=\n-1)*(\n*\n - (2*\n-1-x)*(2*\n-2-x)/2)}; 
          \addlegendentry{$k(r)$}
          
        \end{axis} 
    \end{tikzpicture}
    \quad\quad
    \begin{tikzpicture}[scale=0.75]      

        \pgfplotsset{every tick label/.append style={font=\normalsize}}
        
        \begin{axis}[
          xmin=0,
          xmax=1,
          ymin=0,
          ymax=0.7,
          xlabel={$\gamma$},
          ylabel={radius},
          xtick={0.2, 0.4, ..., 0.8},
          xlabel style={anchor=north, at={(0.5,-0.08), font=\normalsize}},
          ylabel style={anchor=south, at={(-0.1,0.5), font=\normalsize}},
          legend style= {anchor = north east, at={(0.985,0.985), font=\small}},
          cycle list name=mark list*
          ]
          
          \tikzstyle{my_style}=[domain=0:1, mark=none, samples=100]

          \addplot+[my_style, color=red]
          {(1-x)/2}; 
          \addlegendentry{relative unique decoding radius $1-\gamma$}

          \addplot+[my_style, color=black]
          {2-x-sqrt(3-2*x)}; 
          \addlegendentry{relative ECP radius $2 - \gamma - \sqrt{3 - 2\gamma}$}
          
        \end{axis} 
    \end{tikzpicture}
    
    \caption{\label{fig:radii} \small On the left, representation of the minimum distance $d(r)$ and the dimension $k(r)$ of $\RM_\bftheta(r, \bfn)$ depending on $r$,  for $\bfn = (n,n)$. On the right, representation of relative decoding radii of $\RM_\bftheta(\gamma n, \bfn)$ with $n \gg 1$, depending on $\gamma$.}
    \end{figure}
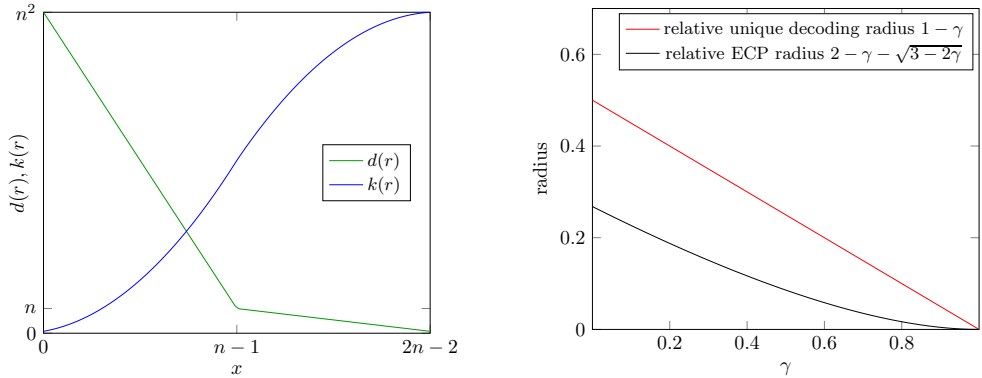

\end{example}

\subsection{Connection with classical Reed--Muller codes}

In this section we prove a relation between $\bftheta$-Reed--Muller
codes and affine cartesian codes in the specific setting where the
base field $\K$ contains all the $n_i$-th roots of unity. For
convenience, we restrict our study to $\bftheta$-Reed--Muller codes of
type $\bfn=(n,\ldots,n) \in \N^m$, for which affine cartesian codes
are classical $q$-ary Reed--Muller codes.
See \cite{geil2013weighted,lopez2014affine} for more
details on affine cartesians codes.

We therefore consider a Galois extension $\LL/\K$ of degree $N = n^m$,
such that
$\Gal(\LL/\K)=\langle \theta_1, \ldots, \theta_m\rangle \cong
(\Z/n\Z)^m$.  Furthermore, we assume that $\LL/\K$ is a Kummer
extension, {hence} $x^n-1$ completely splits in linear factors in
$\K$. Equivalently, $\K$ contains all the $n$--th roots of unity.

We give some additional notation now. Fix $i \in \{1,\ldots, m\}$. The
subgroup
$G_i \mydef \langle \theta_j \mid j \in \{1,\ldots, m\} \setminus
\{i\} \rangle$ yields a fixed field $\LL_i \mydef \LL^{G_i}$, for
$i \in \{1,\ldots, m\}$.  Let us also define
$\E_i \mydef \LL^{\theta_i}$.  We see that $\LL=\LL_i\E_i$,
$\LL_i\cap \E_i=\K$ and $\LL=\LL_1\LL_2\cdots\LL_m$.  Moreover, since
$\LL/\K$ is a Kummer extension and $[\LL_i:\K]=n$, the extension
$\LL_i/\K$ is also a Kummer extension with Galois group
$\Gal(\LL_i/\K)=\langle \theta_i\rangle\cong \Z/n\Z$.
Additionally, for this kind of extensions we have the
following theorem, which is a consequence of the more general abelian
Kummer theory (see \cite[Ch. VI, Sec. 8]{lang2002}).

\begin{theorem}\label{thm:Kummer}
  Let $\LL/\K$ be an abelian extension and $\K$ contains the $n$--th roots
  of unity. If
  $\Gal(\LL/\K)$ has exponent\footnote[1]{A group $G$ is said to
    have exponent $n$ if every element $g\in G$ satisfies
    $g^n=\mathrm{id}$} $n$, then
  $\LL = \K( \sqrt[n]{a_1},\ldots ,\sqrt[n]{a_m} )$ for some
  $a_1,\ldots, a_m \in \K^*$. Conversely, every extension
  $\K( \sqrt[n]{a_1},\ldots ,\sqrt[n]{a_m} )$ is abelian of exponent
  $n$.
\end{theorem}

As a consequence of Theorem \ref{thm:Kummer}, there exist
$a_i \in \K$ and $\aalpha_i \in \LL_i$ such that $\aalpha_i^n=a_i$ and
$\LL_i=\K(\aalpha_i)$.  This implies that the set
$A_i \mydef \{\aalpha_i^j \mid j=0,1,\ldots, n-1\}$ is a $\K$-basis of
$\LL_i/\K$ and
$$A_1\cdot A_2 \cdots A_m \mydef \left\{\prod_{i=1}^m
  \aalpha_i^{j_i} ~\Bigg|~ j_1,\ldots,j_m \in \{0,\ldots, n-1\} \right\}$$
is a $\K$-basis of $\LL/\K$. Furthermore,
$\LL=\K(\aalpha_1,\ldots, \aalpha_m)$ and we have

\begin{equation}\label{eq:thetaalpha}
  \theta_i^s(\aalpha_j^r)= \begin{cases} \aalpha_j^r & \mbox{ if } i \neq j \\
    \zeta_n^{rs}\aalpha_j^r & \mbox{ if } i=j, \end{cases}
\end{equation}
where $\zeta_n \in \K$ is a primitive $n$-th root of unity.  Consider
now for $i \in \{0,\ldots,m\}$ the set $\B_i \mydef A_1\cdots A_i$,
where
$\mathcal{U} \cdot \mathcal{V} = \{ uv, u \in \mathcal{U}, v \in
\mathcal{V} \}$.  By convention, $\B_0 \mydef \{1\}$. Moreover, for
$\bfalpha=(\aalpha_1,\ldots, \aalpha_m)$ and
$\bfi=(i_1,\ldots,i_m) \in \Delta(n)^m$, we write
$\bfalpha^\bfi \mydef \prod_{j=1}^m \aalpha_j^{i_j}$. We consider the
reverse lexicographic order $\prec$ on $\N^m$, from which we reorder
the set $\Delta(n)^m = \{ \bfi_1, \dots \bfi_N \}$.  With this
notation $\B_m=\{\bfalpha^{\bfi_1}, \ldots , \bfalpha^{\bfi_N}\}$, and
for every $t \in \{1,\ldots, m\}$ we have
$\B_t=\{\bfalpha^{\bfi_1},\ldots, \bfalpha^{\bfi_{n^t}}\}$. In
particular, it holds that
\begin{equation}\label{eq:BM} 
\B_m=\bigcup_{j=1}^n\aalpha_m^j\cdot \B_{m-1}. 
\end{equation}

Different bases of $\LL/\K$ produce equivalent vector codes (in the
rank-metric sense).  For this reason, we can restrict our study to
$\mathrm{RM}_{\bftheta,\B}(r,\bfn) \subseteq \LL^N$ for the specific
basis $\B=\B_m$ defined above. We already know that a basis for the
space $\mathrm{RM}_\bftheta(r,\bfn)$ is given by the set
$T_{r,\bfn}=\{\bftheta^\bfi \mid \bfi \in \Delta(\bfn), |\bfi|\leq
r\}$. We define
$\bar{\bftheta} \mydef (\theta_1,\ldots, \theta_{m-1})$ and
$\bar{\bfn} \mydef (n,\ldots, n)\in \N^{m-1}$ and we write

\begin{equation}
  \label{eq:TM} 
T_{r,\bfn}=\bigcup_{j=0}^r \{\bar{\bftheta}^\bfi\theta_m^j
\mid \bfi \in \Delta(\bar{\bfn}),0 \leq  j<n, |\bfi|\leq r-j \}
=\bigcup_{j=0}^r \theta_m^jT_{r-j,\bar{\bfn}},
\end{equation}
where $T_{r,\bar{\bfn}}=\varnothing$ whenever $r<0$.  Furthermore,
for a given $s \in \{0, \dots, m\}$,
we denote by ${\rm Diag}(\B_s)$ the $n^s \times n^s$
diagonal matrix whose entries are given by $\bfalpha^\bfi$, ordered in
the reverse lexicographic order $\prec$.

With this notation, we can now study the generator matrix of the
$k$-dimensional code $\mathrm{RM}_{\bftheta, \B_m}(r,\bfn)$.

\begin{proposition}\label{prop:recursiveGenMatrix}
  Let $G_{r,m}\in \LL^{k \times N}$ be the generator matrix of
  $\mathrm{RM}_{\bftheta, \B_m}(r,\bfn)$ obtained by evaluating the
  $\bftheta$-monomials in $T_{r,\bfn}$. Then
  $G_{r,m}=Y_{r,m}{\rm Diag}(\B_m)$, where
  \begin{enumerate}
  \item\label{part1} If $r=0$, then $Y_{0,m}=(1,1,\ldots, 1)$.
  \item\label{part2} If $m=1$, then 
    $$Y_{r,1}= \begin{pmatrix} 1 & 1 & 1 & \ldots & 1 \\
    1 & \zeta_n &  \zeta_n^2 &  \ldots & \zeta_n^{n-1} \\
    \vdots & \vdots & \vdots & & \vdots \\
    1 & \zeta_n^r & \zeta_n^{2r} & \ldots & \zeta_n^{(n-1)r} 
  \end{pmatrix}.$$
  \item\label{part3} If $r\geq 1$ and $m \geq 2$, then 
    \[
    Y_{r,m}=\begin{pmatrix} Y_{r, m-1} &  Y_{r,m-1} &  Y_{r,m-1} &\ldots &  Y_{r,m-1} \\
      Y_{r-1,m-1} & \zeta_n Y_{r-1,m-1} & \zeta_n^2 Y_{r-1,m-1} & \ldots & \zeta_n^{n-1} Y_{r-1,m-1} \\
      Y_{r-2,m-1} & \zeta_n^2 Y_{r-2,m-1} &  \zeta_n^4 Y_{r-2,m-1} & \ldots & \zeta_n^{2(n-1)} Y_{r-2,m-1} \\
      \vdots & \vdots & \vdots & & \vdots \\
      Y_{0,m-1} & \zeta_n^r Y_{0,m-1} &  \zeta_n^{2r} Y_{0,m-1}&   \ldots &  \zeta_n^{r(n-1)} Y_{0,m-1} \\
    \end{pmatrix}
    \]
  \end{enumerate}
\end{proposition}

\begin{proof}
    \begin{enumerate}
    \item If $r=0$, then
      $\RM_{\bftheta}(0,\bfn)=\Span{\LL}{\Id}$, and
      hence for every ordered basis $\B$ of $\LL/\K$, we have
      $G_{0,m}=(1,\ldots,1){\rm Diag}(\B)$. In particular, it holds
      for $\B_m$.
    \item If $m=1$, then we are in the case of a cyclic Galois group
      $G=\langle \theta \rangle$. It is easy to see by
      \eqref{eq:thetaalpha}, that the action of $\theta$ leads to
      $Y_{r,1}$ being a Vandermonde matrix.
    \item We order the elements in $T_{r,\bfn}$ according to the
      reverse lexicographic order
      $\theta_1\prec \ldots \prec \theta_m$,
      and evaluate them in increasing order.  This leads to a block
      division of $G_{r,m}$, in which the first block of rows
      corresponds to the evaluation of $T_{r,\bar{\bfn}}$, the second
      block of rows to the evaluation of $\theta_mT_{r-1,\bar{\bfn}}$,
      and so on as explained in \eqref{eq:TM}.  Moreover, we have also
      ordered the elements of the basis $\B_m$ according to the
      reverse lexicographic order, which leads to a columns division
      of $G_{r,m}$ in blocks as explained in \eqref{eq:BM}.  The first
      block of columns correspond to $\B_{m-1}$, the second block of
      columns to $\aalpha_m\cdot\B_{m-1}$ and so on.  To sum up, this
      produces a block structure of $G_{r,m}$ in which the
      $(i,j)$-block corresponds to the evaluation of
      $\theta_m^{i-1}T_{r-i+1,\bar{\bfn}}$ in
      $\aalpha_m^{j-1}\cdot \B_{m-1}$.

      Now, by \eqref{eq:thetaalpha} we have
      $\sigma(\aalpha_m)=\aalpha_m$ for every
      $\sigma \in T_{r-i+1,\bar{\bfn}}$. Moreover, it holds that
      $\theta_m(\B_{m-1})=\B_{m-1}$ and
      $
      \theta_m^{i-1}(\aalpha_m^{j-1})=\zeta_m^{(i-1)(j-1)}\aalpha_m$. By
      definition, the matrix associated to
      $T_{r-i+1,\bar{\bfn}}(\B_{m-1})$ is
      $Y_{r-i+1,m-1}{\rm Diag}(\aalpha_m^{j-1}\cdot \B_{m-1})$. Hence,
      the $(i,j)$-block of $G_{r,m}$ is equal to
      \[
        \zeta_m^{(i-1)(j-1)}Y_{r-i+1,m-1}{\rm
          Diag}(\aalpha_m^{j-1}\cdot\B_{m-1}),
        \] which gives the desired
      result.
\end{enumerate}
\end{proof}

As a byproduct we now show that we get a characterization of the
generator matrix $G_{r,m}$ which relates $\bftheta$-Reed--Muller codes
with classical Reed--Muller codes (or affine variety codes or affine
cartesian codes). Consider the set
$$P_{r,m} \mydef \left\{ p \in \K[x_1,\ldots,x_m] \mid \deg p\leq r \right\}.$$
For a finite subset $U\subset \K$ with cardinality $n$, we consider
the set $X \mydef U\times \cdots\times U= U^{m}$, and a total order on
it, such that we can write $X=\{u_1,\ldots, u_{n^m}\}$. Then the
classical Reed--Muller code (or affine variety code, or affine
cartesian code) on $X$ is
\[
  \mathrm{HRM}_{X}(r,m)=\left\{(p(u_1),\ldots, p(u_{n^m})) \mid p \in
    P_{r,m} \right\} \subseteq \K^N.
\]

\begin{theorem}\cite[Proposition 5]{geil2013weighted}\cite[Theorem 3.8]{lopez2014affine}
  If $r\geq 1$ and $U$ has cardinality $n \geq 2$, then the code
  $\mathrm{HRM}_{X}(r,m)$ is an $[N,k,d]_{\K}$ code in the Hamming
  metric, with $N=n^m$ and $d=(n-\ell)n^{m-s-1}$, where $\ell$ and $s$
  are the unique non-negative integers such that $r=s(n-1)+\ell$ and
  $0 \leq \ell <n-1$.
\end{theorem}

We now consider the special case when $U=U_n$ is the set of $n$-th
roots of unity. Every element in $(U_n)^m$ is of the form
$(\zeta_n^{j_1}, \zeta_n^{j_2}, \ldots, \zeta_n^{j_m})=:\bfzeta^\bfj,$
where $\bfj=(j_1,\ldots, j_m) \in \Delta(n)^m$. We order the elements
$\bfzeta^\bfj$'s of $X \mydef U_n^m$ according to the reverse
lexicographic order on $\Delta(n)^m$, and we obtain the following
result.

\begin{theorem}
  The $\bftheta$-Reed--Muller code $\mathrm{RM}_{\bftheta, \B_m}(r,\bfn)$
  has a generator matrix of the form
  $G_{r,\bfn} \mydef Y_{r,m}{\rm Diag}(\B_m),$ where
  $Y_{r,m}\in \K^{k \times N}$ is the generator matrix of the
  classical Reed--Muller codes $\mathrm{HRM}_{X}(r,m)$ obtained by
  evaluating the monomials on the points of $X \mydef (U_n)^m$.
\end{theorem}

\begin{proof}
  The generator matrix for a classical Reed--Muller codes
  $\mathrm{HRM}_{X}(r,m)$ follows the same recursive relations
  described in Proposition \ref{prop:recursiveGenMatrix} part
  \ref{part3}, with the same initial conditions given in \ref{part1}
  and \ref{part2}.
\end{proof}

  In the general case of
  $\Gal(\LL/\K) \cong \Z/n_1\Z \times  \dots \times \Z/n_m\Z$, i.e. for a
  $\bftheta$-Reed--Muller code of type $\bfn=(n_1,\ldots, n_m)$, a
  similar result can be shown. More specifically, let
  $V_i \mydef \{x \in \K \mid x^{n_i}=1\}$ and
  $X \mydef V_1\times \cdots \times V_m$. Then the code
  $\mathrm{RM}_{\bftheta, \B_m}(r,\bfn)$ has a generator matrix which
  is equal to the generator matrix of the code $\mathrm{HRM}_X(r,m)$
  multiplied on the right by ${\rm Diag}(\B)$, where $\B$ is the ordered
  $\K$-basis of $\LL/\K$ with respect to the reverse lexicographic
  order, which is constructed as explained for the case
  $n_1=\cdots=n_m=n$.

\section{Conclusion and open problems}

In this paper was presented a general description of codes seen as
subspaces of the group algebra $\LL[G]$ with arbitrary Galois
extensions $\LL/\K$. Analogues of Reed--Muller codes were constructed
as an application, but there remains some way to go towards
practicality of these codes.

First, one can wonder whether $\bftheta$-Reed--Muller codes can be
decoded up to hald their minimum distance. Such decoding algorithms
are known for Hamming-metric Reed--Muller codes over finite
fields. However they require to embed the code in a Reed--Solomon code
over the extension field $\LL$, and to use the decoder attached to
this code. To our opinion, this technique seems difficult to adapt in
our context, given the fact that there is no way to embed a
$\bftheta$-Reed--Muller code into a Gabidulin code (since $G$ is not
cyclic).

Second, the lack of practicality of our codes relies on the fact that,
if $\LL/\K$ is not cyclic, then $\LL$ cannot be a finite field. This
raises the two following issues: (i) find Galois extensions $\LL/\K$
in which computations are efficiently doable (\emph{so-called}
effective fields), and (ii) find maps $\pi : \LL \to \F$, where $\F$
is an effective field, such that $\pi$ sends a code
$\C \subseteq \LL[G]$ to a \enquote{good} code
$\pi(\C) \subseteq \F^n$ whose properties can be derived from those of
$\C$.

\section*{Acknowledgements}

The authors would like to thank the organizers of Dagstuhl seminar
no.\ 18511 \enquote{Algebraic Coding Theory for Networks, Storage, and
  Security} where was initiated this research project. This project
was also partially funded by French grant no.\ ANR-15-CE39-0013
\enquote{Manta} which enabled the authors to meet during a workshop
held at Nogaro, France.

J. Lavauzelle is funded by French \emph{Direction G{\'e}n{\'e}rale
  l'Armement}, through the \emph{P{\^o}le d'excellence cyber}.

A. Neri is funded by \emph{Swiss National Science Foundation}, through
grant no. 187711.  \bibliographystyle{abbrv} \bibliography{biblio}

\appendix

\section{A second proof for the minimum distance lower bound}
The algebra $\LL[G]$
can also be represented as a skew polynomial ring modulo a particular
two--sided ideal. Let us recall that the skew polynomial ring
$\LL[\bfx; \bftheta] = \LL[x_1, \dots, x_m; \theta_1, \dots,
\theta_m]$ is the ring of polynomials $Q(\bfx) = Q(x_1, \dots, x_n)$
where the addition is defined as in the usual polynomial ring, and the
multiplication follows the following rules
\[
\begin{array}{rcll}
  x_ix_j &=&x_jx_i &\quad  \mbox{ for any } i,j \in \{1,\ldots, m\}, \\
  x_ia &=&\theta_i(a)x_i  &\quad \mbox{ for any } a \in \LL,
\end{array}
\]
and is extended by associativity. 
It is known that the center of this ring is $\K[\bfx^{\bfn}]$, and the ideal generated by $(x_1^{n_1}-1, x_2^{n_2}-1, \ldots, x_m^{n_m}-1)$ is two-sided. We will indicate such ideal by $I_{\bfn}$.

The ring $\LL[\bfx;\bftheta]$ is a very particular case of left Poincar\'e-Birkhoff-Witt ring, for which the theory of Gr\"obner basis is well-defined and it works practically in the same way as for commutative rings. For a deeper understanding on the topic, we refer the interested reader to \cite{bueso2003algorithmic}.

\begin{theorem}
  \label{thm:isomorphism-skew-polynomials-group-algebra}
  Let
  $G \mydef \Gal(\LL/\K)=\langle \theta_1,\ldots, \theta_m \rangle
  \cong \Z/n_1\Z\times \cdots \times \Z/n_m\Z$. Then the map
  \[\Phi: \left\{
  \begin{array}{ccl}
    \LL[\bfx;\bftheta] & \longrightarrow & \LL[G] \\
          \sum_{\bfi \in \N^{m}} \bb_\bfi \bfx^\bfi & \longmapsto &
          \sum_{\bfi \in \N^m} \bb_\bfi \bftheta^\bfi 
  \end{array}\right.
  \]
  is a surjective ring homomorphism with
  $\ker \Phi=I_{\bfn}$. In
  particular, it induces an isomorphism
  $\bar{\Phi}: \LL[\bfx;\bftheta]/I_{\bfn}
  \rightarrow \LL[G]$.
\end{theorem}

With this framework in mind, we propose a second proof of the lower
bound on the rank of a nonzero $\bftheta$-polynomial $P$ given in
Theorem~\ref{thm:AF}. Precisely, we will prove the following: if
$P \in \RM_\bftheta(r, \bfn)$, then
\[
  \rk_\K(P) \ge \min \left\{ \prod_{i=1}^m (n_i -u_i) \;\Big|\; \bfu =
    (u_1, \dots, u_m) \in \Delta(\bfn), |\bfu| \le r \right\}.
\]

\begin{proof}[Second proof:]

  Let
  $P=\sum_{\bfi \in \Delta(\bfn)}\bb_\bfi\bftheta^\bfi \in
  \mathrm{RM}_\bftheta(r, \bfn)$ be a $\bftheta$-polynomial.  Observe
  again that the minimum is attained for a $\bftheta$-polynomial of
  $\bftheta$-degree equal to $r$, and we set
  \[
  \delta \mydef \min\left\{ \prod_{i=1}^m (n_i-u_i) \mid \bfu
  = (u_1,\ldots, u_m)\in \Delta(\bfn), |\bfu|= r\right\}.
  \]
  Therefore, we need to prove that $\w_I(P) \ge \delta$, where
  $\w_I(P)=\dim_\LL(\LL[G]/\Ann_{\LL[G]}(P))$.
  Equivalently, we have to show that that there exists an
  $\LL$-subspace $T$ of $\LL[G]$ of dimension at least $\delta$ such
  that $T \cap \Ann_{\LL[G]}(P) =\{0\}$.  Consider the isomorphism
  $\bar{\Phi}:\LL[\bfx;\bftheta]/I_{\bfn} \rightarrow \LL[G]$
  introduced in
  Theorem~\ref{thm:isomorphism-skew-polynomials-group-algebra}.  Using this isomorphism, our
  goal is equivalent to finding an $\LL$-subspace $V$ of
  $\LL[\bfx;\bftheta]/I_{\bfn}$ of dimension at least $\delta$, such
  that $g(\bfx) \bar{\Phi}^{-1}(P)(\bfx) \neq 0 \mod I_{\bfn}$ for
  every $g(\bfx) \in V$.

  First, we observe that the set $\{x_1^{n_1}-1,\ldots, x_m^{n_m}-1\}$
  is a universal Gr{\" o}bner basis for the ideal $I_{\bfn}$.
 We
  choose the representative $\bar{P}(\bfx)\in \LL[\bfx;\bftheta]$ of
  $\bar{\Phi}^{-1}(P)(\bfx)$ reduced modulo the Gr{\" o}bner basis
  $\{x_1^{n_1}-1,\ldots, x_m^{n_m}-1\}$, that is
  $\bar{P}(\bfx)=\sum_{\bfi \in \Delta(\bfn)}\bb_\bfi \bfx^\bfi \in
  \LL[\bfx;\bftheta]$.
  Moreover, we fix a monomial order $\prec$, and we consider the
  leading term of $\bar{P}(\bfx)$ with respect to $\prec$, that is
  $\mathrm{lt}_{\prec}(\bar{P}(\bfx))=\bfx^\bfu$, for
  $\bfu = (u_1,\ldots, u_m)$, and we consider the set
  \[
    Z=\{f(\bfx) \in \LL[\bfx;\bftheta] \mid \deg_{x_i}(f) < n_i-u_i,\,
    i=1, \dots, m\}.
  \]
  Note that $Z\cap I_{\bfn}=\{0\}$. This is due to the fact that the
  set $\{x_1^{n_1}-1,\ldots, x_m^{n_m}-1\}$ is a universal Gr{\"
    o}bner basis for the ideal $I_{\bfn}$ and none of the monomials in
  $Z$ belongs to monomial ideal spanned by the leading terms of the
  generators of $I_\bfn$, namely
  $\mathrm{lt}_{\prec}(I_{\bfn})=(x_1^{n_1}, \ldots,
  x_m^{n_m})$. Therefore, the canonical projection
  $\pi: \LL[\bfx;\bftheta] \rightarrow \LL[\bfx;\bftheta]/I_{\bfn}$ is
  injective when restricted to $Z$.

  At this point let us take an arbitrary skew polynomial
  $f(\bfx) \in Z$ and consider its leading term
  $\mathrm{lt}_\prec(f(\bfx))=\bfx^\bfv$, where, by definition of the
  space $Z$, we have $\bfv = (v_1,\ldots, v_m)$ and $v_i <n_i-u_i$ for
  all $i=1, \dots, m$.  Then,
  \[
    \mathrm{lt}_\prec(f(\bfx)\bar{P}(\bfx))=\mathrm{lt}_\prec(f(\bfx))\mathrm{lt}_\prec(\bar{P}(\bfx))=
    \bfx^\bfv\bfx^\bfu=\bfx^{\bfu+\bfv}.
  \]
  Since $u_i+v_i<n_i$ for every $i$, we have that
  $\mathrm{lt}_\prec(f(\bfx)\bar{P}(\bfx)) \notin (x_1^{n_1},\ldots,
  x_m^{n_m})= \mathrm{lt}_\prec(I_{\bfn}) $. Therefore,
  $f(\bfx)\bar{P}(\bfx) \notin I_{\bfn}$. Denote by
  $\pi:\LL[\bfx;\bftheta]\to\LL[\bfx;\bftheta]/I_{\bfn}$ the canonical
  projection modulo the ideal $I_{\bfn}$. Hence,
  $\pi(f(\bfx))\pi(\bar{P}(\bfx))=\pi(f(\bfx))\bar{\Phi}^{-1}(P)\neq
  0$.
  Thus, the space $V \mydef \pi(Z)$ is such that
  $g(\bfx)\bar{\Phi}^{-1}(P)(\bfx) \neq 0 \mod I_{\bfn}$ for every
  $g(\bfx) \in V$.  Moreover,
  $$\dim_{\LL}(\pi(Z))=\dim_{\LL}(Z)=\prod_{i=1}^m(n_i-u_i),$$
  which concludes the proof.
\end{proof}

\end{document}